\documentclass[11pt,letterpaper]{article}
\usepackage[utf8]{inputenc}
\usepackage[T1]{fontenc}
\usepackage{lmodern}
\usepackage[margin=1in]{geometry}
\usepackage{enumitem}
\usepackage{microtype}

\usepackage{amssymb}
\usepackage{amsmath}
\usepackage{amsthm}
\usepackage{thmtools}
\usepackage{thm-restate}
\usepackage{tikz}

\usepackage[procnumbered,ruled,vlined,linesnumbered]{algorithm2e}

\usepackage{xcolor}
\usepackage{xspace}
\usepackage{xfrac}
\usepackage{algorithmic} 
\usepackage[colorlinks]{hyperref}
\usepackage[capitalize]{cleveref}
\usepackage[draft]{fixme}

\newcommand{\floor}[1]{\lfloor #1 \rfloor}

\fxsetup{mode=multiuser,theme=color, layout=inline}
\FXRegisterAuthor{dd}{add}{\color{red}{\bf{Debarati}}}
\FXRegisterAuthor{tk}{atk}{\color{blue}{\bf{Tomasz}}}
\FXRegisterAuthor{bs}{abs}{\color{violet}{\bf{Barna}}}

\usepackage{authblk}

\usepackage{comment}

\newtheorem{theorem}{Theorem}[section]
\newtheorem*{theorem*}{Theorem}

\newtheorem{lemma}[theorem]{Lemma}
\newtheorem{claim}[theorem]{Claim}

\newtheorem{proposition}[theorem]{Proposition}
\newtheorem{observation}[theorem]{Observation}

\newtheorem{fact}[theorem]{Fact}
\crefname{fact}{Fact}{Facts}
\theoremstyle{definition}
\newtheorem{definition}[theorem]{Definition}
\theoremstyle{remark}

\newcommand{\W}{\mathcal{W}}
\newcommand{\Oh}{O}
\newcommand{\tOh}{\tilde{\Oh}}
\newcommand{\sub}{\subseteq}
\newcommand{\cost}{\mathsf{cost}}
\newcommand{\D}{\mathcal{D}}
\newcommand{\dd}{\mathinner{.\,.}}
\newcommand{\sm}{\setminus}
\newcommand{\rev}[1]{\overline{#1}}
\newcommand{\DYCK}{\mathsf{Dyck}}
\newcommand{\dyck}{\mathsf{dyck}}

\title{Improved Approximation Algorithms for Dyck Edit Distance\\and RNA Folding}
\author[1]{Debarati Das\thanks{Supported by Basic Algorithms Research Copenhagen (BARC), grant 16582 from the VILLUM Foundation.}}
\author[2]{Tomasz Kociumaka\thanks{Partially supported by NSF 1652303, 1909046, and HDR TRIPODS 1934846 grants, and an Alfred P. Sloan Fellowship.}}
\author[2]{Barna Saha\protect\footnotemark[2]}

\affil[1]{University of Copenhagen, Denmark}
\affil[ ]{\texttt{debaratix710@gmail.com}}
\affil[2]{University of California, Berkeley, USA}
\affil[ ]{\texttt{\{kociumaka,barnas\}@berkeley.edu}}

\date{}

\begin{document}

\maketitle
\begin{abstract}
The Dyck language, which consists of well-balanced sequences of parentheses,
 is one of the most fundamental context-free languages. 
The Dyck edit distance quantifies the number of edits (character insertions, deletions, and substitutions) required to make a given length-$n$ parenthesis sequence well-balanced. RNA Folding involves a similar problem, where a closing parenthesis can match an opening parenthesis of the same type irrespective of their ordering. For example, in RNA Folding, both $\texttt{()}$ and $\texttt{)(}$ are valid matches, whereas the Dyck language only allows $\texttt{()}$ as a match. Both of these problems have been studied extensively in the literature. Using fast matrix multiplication, it is possible to compute their exact solutions in time $\Oh(n^{2.824})$ (Bringmann, Grandoni, Saha, V.-Williams, FOCS'16), and a $(1+\epsilon)$-multiplicative approximation is known with a running time of $\Omega(n^{2.372})$.

The impracticality of fast matrix multiplication often makes combinatorial algorithms much more desirable. Unfortunately, it is known that the problems of (exactly) computing Dyck edit distance and folding distance are at least as hard as Boolean matrix multiplication.
Thereby, they are unlikely to admit truly subcubic-time combinatorial algorithms. In terms of fast approximation algorithms that are combinatorial in nature, the state of the art for Dyck edit distance is a $\Oh(\log n)$-factor approximation algorithm that runs in near-linear time (Saha, FOCS'14), whereas for RNA Folding only an $\epsilon n$-additive approximation in $\tOh(\frac{n^2}{\epsilon})$ time (Saha, FOCS'17) is known. 

In this paper, we make substantial improvements to the state of the art for Dyck edit distance (with any number of parenthesis types). We design a constant-factor approximation algorithm that runs in $\tOh(n^{1.971})$ time (the first constant-factor approximation in subquadratic time). Moreover, we develop a $(1+\epsilon)$-factor approximation algorithm running in $\tOh(\frac{n^2}{\epsilon})$ time,
	which improves upon the earlier additive approximation. Finally, we design a $(3+\epsilon)$-approximation that takes $\tOh(\frac{nd}{\epsilon})$ time, where $d\ge 1$ is an upper bound on the sought distance.
	
As for RNA folding, for any $s\ge 1$, we design a multiplicative $s$-approximation algorithm that runs in  $O(n+\left( \frac{n}{s} \right)^3)$ time. To the best of our knowledge, this is the first nontrivial approximation algorithm for RNA Folding that can go below the $n^2$ barrier. All our algorithms are combinatorial in nature. 
\end{abstract}
\thispagestyle{empty}
\setcounter{page}{0}
\newpage

\section{Introduction}
The Dyck language is a well-known context free language consisting of well-balanced sequences of parentheses. Ranging from programming syntaxes, to  arithmetic and algebraic expressions, environments in LaTeX, and tags in HTML/XML documents -- we observe instances of Dyck language everywhere. For a comprehensive discussion on Dyck language and context free grammars see~\cite{h:78,k:12}. Given a sequence $x$ of $n$ parentheses (which may be unbalanced), the \emph{Dyck edit distance problem} asks for the minimum number of edits (character insertions, deletions, and substitutions) needed to make $x$ well-balanced. Interestingly, string edit distance, which is one of the fundamental string similarity measures, can be interpreted as a special case of Dyck edit distance\footnote{Given two strings $s$ and $t$, form a sequence of parentheses by concatenating $s$, interpreted as a sequence of opening parentheses, and the reverse complement of $t$, obtained by reversing $t$ and replacing each symbol with the corresponding closing parenthesis, not present in the original alphabet.}.

A simple dynamic programming can compute Dyck edit distance in $O(n^3)$ time. After nearly four decades, this cubic running time was improved in 2016 by Bringmann, Grandoni, Saha, and V. Williams~\cite{BGSW19}, who gave the first sub-cubic exact algorithm for a more general problem of language edit distance~\cite{ap72}. The algorithm uses not-so-practical fast Boolean matrix multiplication, arguably so because computing Dyck edit distance is at least as hard as Boolean matrix multiplication~\cite{abw:15}, and hence combinatorial subcubic algorithms are unlikely to exist. 
%The Dyck edit distance is a significant generalization of the well-studied string edit distance problem, for which exact quadratic time algorithms follow from easy dynamic programming.

A problem closely related to Dyck edit distance is RNA Folding~\cite{nussinov1980fast}. Both in RNA Folding and Dyck edit distance, parentheses must match in an uncrossing way. However, in a RNA-folding instance, a close parenthesis of same type can match an open parenthesis irrespective of the order of occurrences. For example, under the RNA Folding distance, both $\texttt{()}$, and $\texttt{)(}$ are valid matches, whereas the Dyck language only allows $\texttt{()}$ as a match. In terms of exact computation, they exhibit the same time complexity~\cite{BGSW19,abw:15}\footnote{In these problems, we are aiming to minimize the number of non-matched parentheses as opposed to maximizing the matched parentheses.}.

Can we design fast approximation algorithms for Dyck edit distance and RNA Folding? The first progress on this question for Dyck edit distance was made by Saha~\cite{s:14}, who proposed a polylogarithmic-factor approximation algorithm that runs in near linear time. It is possible to  provide an $\epsilon n$ additive approximation for any $\epsilon >0$ in $\tOh(\frac{n^2}{\epsilon})$ time~\cite{s:17}. Therefore, unless the distance is $O(n)$, and we allow quadratic time, the above algorithm does not provide a constant factor approximation to Dyck edit distance. This latter result on additive approximation applies to RNA Folding as well.  Backurs and Onak~\cite{BO16} showed an exact algorithm for Dyck edit distance that runs in $O(n+d^{16})$ time which was recently improved by Fried, Golan, Kociumaka, Kopelowitz, Porat and Starikovskaya \cite{otherSubmission} to run in $O(n+d^5)$ time (and $\tOh(n+d^{4.783})$ using fast matrix multiplication). Therefore, as the state of the art stands, (i) there did not exist any nontrivial multiplicative approximation for RNA Folding in subquadratic time, and (ii) there did not exist any subquadratic-time constant factor approximation for Dyck edit distance that works for the entire distance regime.

Let us contrast this state of affairs with the progress on string edit distance approximation. As mentioned earlier, string edit distance is a special case of Dyck edit distance. Early work~\cite{lms:98,bjkk:04,bfc:06,ak:09} on approximating string edit distance resulted in the first near-linear time polylogarithmic-factor approximation in 2010 by Andoni, Krauthgamer, and Onak~\cite{ako:10}. It took another
eight years to obtain the first constant-factor approximation of edit distance in subquadratic time~\cite{CDGKS18} (see~\cite{BEGHS21} for a quantum analog). Finally, Andoni and Nosatzki improved the running time to near-linear while maintaining a constant factor approximation~\cite{AN20}. Using the best result in string edit distance approximation~\cite{AN20}, it is possible to improve the approximation factor of~\cite{s:14} to $O(\log{n})$. However, designing a constant factor approximation for Dyck edit distance in subquadratic time remains wide open. RNA Folding even though conceptually very similar to Dyck edit distance, the algorithm of \cite{s:14} does not apply to it. 

Saha's work on Dyck approximation~\cite{s:14} developed a random walk technique which has later been used for edit distance embedding and document exchange~\cite{DBLP:conf/stoc/ChakrabortyGK16,DBLP:conf/focs/BelazzouguiZ16}. Using this random walk, it is possible to  decompose a parenthesis sequence into many instances of string edit distance problem. However, this decomposition loses a logarithmic factor in the approximation, raising the question whether there exists an efficiently computable decomposition with a significant less loss.

\paragraph*{Contributions.}
We provide several efficient algorithms approximating Dyck edit distance, and RNA Folding.
\paragraph*{Faster Approximation Algorithms for Dyck Edit Distance}
\begin{itemize}
\item 
\textbf{Constant approximation in Subquadratic time.}
\emph{The main contribution of this paper is the first constant factor approximation algorithm for Dyck edit distance that runs in truly subquadratic time.} 
The running time of our algorithm is $\tOh(n^{1.971})$. We remark that to keep the algorithm and the analysis simple we do not optimize the exponent in the running time.
We employ and significantly extend the tools that have been previously developed in connection to string edit distance and related measures such as the windowing strategy, window-to-window computation, sparse and dense window decomposition etc.~\cite{CDGKS18,BEGHS21,DBLP:conf/stoc/GoldenbergRS20}. These methods are tied to problems involving two or more strings (unlike the Dyck edit distance, which is a single sequence problem). Given the universality of Dyck edit distance, the techniques developed in this paper may lead to further advancements for more generic problems like the language edit distance etc.~\cite{s:15,s:17,BGSW19}.

Our main algorithm handles the cases of large and small Dyck edit distance separately. 
    \item \textbf{Small Dyck edit distance.}
When the Dyck edit distance $d$ is small, we give a $(3+\epsilon)$-approximation algorithm that runs in $\tOh(\frac{nd}{\epsilon})$ time. We can contrast this result to the previously known time bound of $O(n+d^{16})$ by Backurs and Onak~\cite{BO16} which improves to $O(n+d^6)$ if substitutions are not allowed. 
In a parallel work~\cite{otherSubmission}, these running times were improved to $O(n+d^5)$ (combinatorially)
and $\tOh(n+d^{4.783})$ (using randomization).
Nevertheless, even in a hypothetical bast-case scenario that a combinatorial $O(n+d^3)$-time algorithm exists, an $\tOh(nd)$-time algorithm is still faster for all $d \gg \sqrt{n}$. 

\item \textbf{A Quadratic PTAS.} We give a $(1+\epsilon)$ approximation algorithm for Dyck edit distance that runs in $\tOh(\frac{n^2}{\epsilon})$ time. This improves upon the previous result of~\cite{s:17} that gets such a result only when $d=\Theta(n)$. Prior work that achieves multiplicative $(1+\epsilon)$ approximation uses fast Boolean matrix multiplication and has super-quadratic running time~\cite{s:15}. 
\end{itemize}

\paragraph*{RNA Folding}
For RNA Folding, we are aiming to minimize the number of non-matched characters. We often refer to this as the folding distance. For any $s >1$, we give an $s$ multiplicative approximation of RNA Folding
distance in time $O(n+\left(\frac{n}{s}\right)^3)$. This is the first result to our knowledge that goes below the quadratic running time.
%Our subquadratic time algorithm is build upon the constant approximation edit distance algorithm of~\cite{CDGKS18}. 
It is to be noted here that the machinery developed in \cref{sec:consub}, and a new triangle property for Dyck distance (\cref{lem:triangle} in \cref{sec:prelim}) equally apply to the RNA folding problem as well. This implies a  constant-factor approximation of RNA folding in $\tOh(n^{1.97})$ time when the distance is large, say higher than $n^{0.99}$. 

\paragraph*{Discussion and Open Problems.} 
The resemblance between Dyck and string edit distance has already been studied in the literature. As mentioned earlier, the decomposition obtained by the random walk technique 
%showed using a random walk technique an input string $x$ can be decomposed into multiple subsequences where each subsequence comprises of either open or close parentheses. Following this the subsequences then can be glued together using the string edit distance algorithm to construct a matching.
%However this decomposition can 
ensures only an $O(\log n)$ approximation~\cite{s:14}. In this work, instead of reducing the Dyck edit distance to string edit distance, we try to find a direct decomposition of the sequence $x$ into different substrings, where for each substring there is a peer such that they are matched by some optimal alignment (with some error leading to a constant-factor approximation). However unlike the string counterpart, Dyck edit distance does not have the structural property that if an optimal alignment matches the characters of a substring $s_1$ with the characters of a substring $s_2$, then the lengths of $s_1$ and $s_2$ are roughly the same (see Fig~\ref{fig:fig1}). Thus, in our decomposition, the substrings can have varied lengths. In fact, it turns out that if the Dyck edit distance is truly sublinear (i.e., $n^{1-\epsilon}$), then we need to consider roughly $n^\epsilon$ different length substrings to ensure a constant-factor approximation. We remark that this  is one of the  barriers in further pushing down the running time from subquadratic to near linear. We also note that if an analog of our $\tOh(nd)$-time algorithm can be provided for RNA Folding, then we would also get a constant-factor subquadratic algorithm for RNA folding for all distance regimes.

The Dyck recognition problem has been studied extensively in different models including the streaming~\cite{mmn:stoc10,kls:mfcs11,cckm:focs10} and property testing framework~\cite{alon2001regular,prr:rand03,FMS18}. 
However, no sublinear-time approximation algorithms exist either for Dyck edit distance or  RNA Folding. Our algorithm for RNA Folding in this paper (which also applies to Dyck edit distance) runs in $O(n+\left(\frac{n}{s}\right)^3)$ time, and requires a linear time pre-processing step that eliminates pairs of matching adjacent characters, which leaves strongly structured instances. This preprocessing step is currently the main bottleneck in going in the sublinear-time setting.

%We cannot afford this in the sublinear-time setting, so very new ideas are likely necessary.

%For a comprehensive discussion on Dyck language and context free grammars see~\cite{h:78,k:12}.
%\paragraph*{Outline of \cref{sec:PTAS,sec:3apx}}

%\paragraph*{Other Related Work}

\subsection{Technical Overview}
%\begin{itemize}
%\item THESE WILL GO INTO TECHNICAL OVERVIEW.
%    \item Main structural difference: Windows of very different lengths can be very similar.
%   \item We develop a standalone $nd$-time algorithm.
%    \item In the complementary algorithm for large $d$, we are loosing two $d/n$ factors due to the scaling issue (to matching windows of different lengths), and we decided not to optimize further 
 %   \item AN: metric on $d$-size to $d'$-size windows for $d'>d$. Here, the information about size-$d$ windows is not sufficient to reason about size-$d'$ windows.
%\end{itemize}

\begin{figure}[htp]
    \centering
    \includegraphics[scale=0.4]{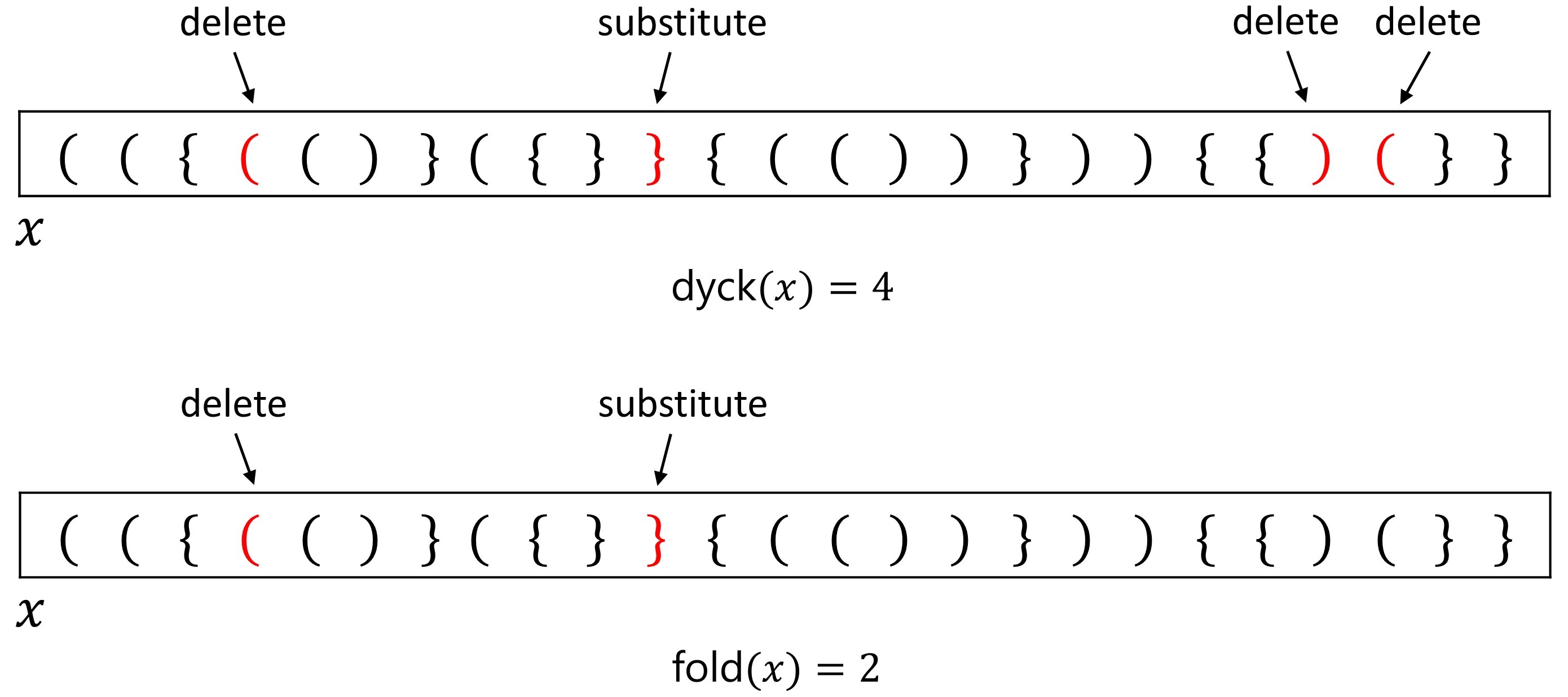}
    \caption{Example of Dyck and folding edit distance. }\label{fig:dyckexmp}
\end{figure}

As input, we are given a string $x$ of length $n$ over alphabet $\Sigma$ that consists of two disjoint sets $T$ and $\rev{T}$ of opening and closing parentheses respectively. The Dyck edit distance of $x$ is the minimum number of parentheses insertions, deletions, and substitutions required to make $x$ well-parenthesized.

\paragraph*{Quadratic-time PTAS.}
The standard $\Oh(n^3)$-time algorithm for the Dyck edit distance is a dynamic-programming procedure that computes the distance of each substring of the input string. The bottleneck of this approach is that, to compute the distance of each substring $x(i\dd j]$, starting at index $i+1$ and ending at index $j$, one needs to iterate over all possible decompositions of the substring into a prefix $x(i\dd k]$ and a suffix $x(k\dd j]$ where the intermediate index $k\in (i\dd j)$ (this corresponds to the fact that the concatenation of two well-parenthesized expression is a well-parenthesized expression).\footnote{For $i,j\in \mathbb{Z}$, we denote $[i\dd j]=\{k\in \mathbb{Z} : i \le k \le j\}$, $[i\dd j)=\{k\in \mathbb{Z} : i \le k < j\}$, $(i\dd j]=\{k\in \mathbb{Z} : i < k \le j\}$, and $(i\dd j)=\{k\in \mathbb{Z} : i < k < j\}$.} We call index $k$ a pivot corresponding to a decomposition. The $\tOh(\frac{n^2}{\epsilon})$-time $\epsilon n$ additive approximation of~\cite{s:17} reduces the number of considered pivots to $\tOh(\frac{1}{\epsilon})$; thus $\tOh(\frac{n}{\epsilon d})$ ($d=\dyck(x)$) different pivots would be necessary for a $(1+\epsilon)$ multiplicative approximation (which is same as $\epsilon d$ additive approximation).
On the other hand, a simple $\Oh(n^2d)$-time algorithm recently developed in~\cite{otherSubmission} is based on a combinatorial observation that $\Oh(d)$ pivots are sufficient after the $\Oh(n)$-time preprocessing from~\cite{BO16,s:14}. 
We start with a brief overview of this algorithm. For any index $i\in [0\dd n]$, we define the height of $i$ to be $h(i)=|\{j\in [1\dd i] : x[j]\in T\}| - |\{j\in [1\dd i] : x[j]\in \rev{T}\}|$ i.e., the difference between the number of opening and closing parentheses in substring $x[1\dd i]$.
An index $v$ is called a valley if $h(v-1)>h(v)<h(v+1)$, i.e., $x[v]$ is a closing parenthesis whereas $x[v+1]$ is an opening parenthesis. \cite{BO16} showed $x$ can be preprocessed in linear time to generate another string $x'$ such that $\dyck(x)=\dyck(x')$ and $x'$ has at most $2d$ %\tknote{I would introduce $d$ a couple of sentences earlier}) 
valleys. \cite{otherSubmission} further claimed that, without loss of generality, it is enough to consider pivots that are at distance $0$ or $1$ from a valley (we henceforth denote the set of such pivots by $K$) plus $\Oh(1)$ pivots next to the boundary of the considered range $(i\dd j)$; note observation yields a $\Oh(n^2d)$-time algorithm. 

In Section~\ref{sec:PTAS}, we provide an algorithm that further restricts the set of pivots $K$ considered for each range $(i\dd j)$ and provides a $(1+\epsilon)$ approximation of $\dyck(x)$ in $\tOh(\epsilon^{-1}n^2)$ time (Theorem~\ref{thm:approxdyck}).
This is inspired by how~\cite{s:17} considered only $\tOh(\epsilon^{-1})$ pivots out of each range $(i\dd j)$. The original argument relies on two observations: that using pivot $k'$ instead of $k$ incurs at most $\Oh(|k-k'|)$ extra edit operations,
and that, for an $\epsilon n$ additive approximation, we can afford $\Oh(\frac{\epsilon \min(j-k,\, k-i)}{\log n})$ extra operations when using pivot $k\in (i\dd j)$. In our multiplicative approximation, we refine the second observation by replacing $\min(j-k,k-i)$ with $\min(|K\cap (j\dd k)|, |K\cap (i\dd k)|)$. 
On the other hand, the first observation is not useful because the set $K\cap(i\dd j)$ is already relatively sparse. Thus, instead of restring each range $(i\dd j)$ to use few pivots $k$, we restrict each pivot $k$ to be used within few ranges $(i\dd j)$. This is feasible with respect to the approximation ratio because the costs for $x(i\dd j]$ and $x(i'\dd j']$ may only differ by $\Oh(|i-i'|+|j-j'|)$, and because the universe of ranges remains dense in the $\Oh(n^2 d)$-time algorithm (unlike the universe of pivots).

\paragraph*{Constant-factor approximation in $\tOh(nd)$ time.}
Overcoming the $\Oh(n^2)$ barrier with a DP approach poses significant challenges: there are $\Theta(n^2)$ substrings to consider and, for $d\ge \sqrt{n}$, this quantity does not decrease (in the worst case)
even if we run the preprocessing of~\cite{BO16,s:14} and exclude substrings with costs exceeding $d$.
Thus, we artificially restrict the DP to substrings whose all prefixes have at least as many opening parentheses as closing ones and whose all suffixes have at least as many closing parentheses as opening ones. Surprisingly, as shown in \cref{sec:3apx}, this yields a 3-approximation of the original cost.
Furthermore, if we additionally require that the number of opening parentheses and the number of closing parentheses across the entire substring are within $2d$ from each other (otherwise, the Dyck edit distance trivially exceeds the threshold), we end up with $\Oh(nd)$ substrings.
Reusing the pivot sparsification of \cref{sec:PTAS}, they can be processed in $\tOh(\epsilon^{-1}nd)$ total time
at the cost of increasing the approximation ratio from $3$ to $3+\epsilon$.

\begin{figure}[htp]
    \centering
    \includegraphics[scale=0.4]{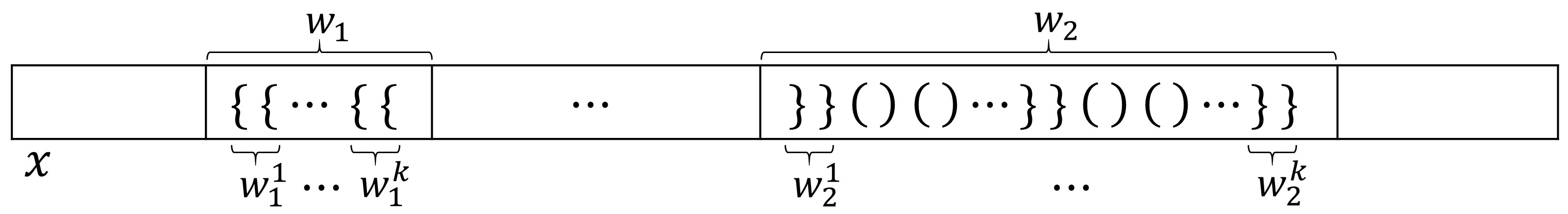}
    \caption{An example showing two very different length substrings can be matched with cost $0$. }\label{fig:fig1}
\end{figure}

\paragraph*{Constant approximation in $\tOh(n^{1.971})$ time.}
In Section~\ref{sec:consub}, we provide an algorithm Dyck-Est that provides a constant approximation in time $\tOh(n^{1.971})$. % thus proving Theorem~\ref{thm:mainapproxsubquadratic}. 
    At a high level, the framework used in our algorithm is similar to the three-step framework used in~\cite{CDGKS18} that provides constant approximation of string edit distance in subquadratic time. %The high level objective here is to first partition the problem in smaller subproblems and then once these subproblems are solved, a dynamic program technique is used to combine them to form a solution of the main problem. 
Thus, we start by providing a brief recap of~\cite{CDGKS18}, and point out the major bottlenecks for applying this framework directly to our problem. Given two strings $x,y$ of length $n$, the algorithm of~\cite{CDGKS18} starts by constructing a set of windows $\mathcal{W}_x$ for $x$ and a set of windows $\mathcal{W}_y$ for $y$ where each window is a length-$d$ subinterval of $[1\dd n]$, representing a substring of $x$ (or of $y$). The motivation behind this construction is the following: provided the edit distance between any pair of windows from  $\mathcal{W}_x$ and $\mathcal{W}_y$, one can compute a constant approximation of $ED(x,y)$ (edit distance of $x,y$) using a dynamic program algorithm in time $O(\frac{n^2}{d^2})$. The challenge here is that if for all pair of windows from  $\mathcal{W}_x$ and $\mathcal{W}_y$, the edit distance is computed using a trivial dynamic program algorithm then the total running time becomes quadratic.~\cite{CDGKS18} managed to show that in some favorable situation, one can use random sampling to select a subset of window pairs from $\mathcal{W}_x\times \mathcal{W}_y$ such that evaluating the edit distance of the window pairs in the subset is enough to construct a solution close to an optimal alignment of $x,y$. 
In the other extreme they showed instead of computing edit distance for each window pairs explicitly, one can use triangle inequality to get a constant approximation of the optimal cost. The use of triangle inequality was first proposed by~\cite{BEGHS21} except they replaced the random sampling by a quantum counterpart; thus providing a quantum algorithm computing constant approximation of edit distance in subquadratic time. Several other works have subsequently used this framework to solve related problems~\cite{DBLP:conf/stoc/GoldenbergRS20,BSS20,RSSS:19}.

Using random walk~\cite{s:14} provides an algorithm that partitions string $x$ into a set of disjoint subsequences where each subsequence starts with a single sequence of open parentheses and ends with a single sequence of closing parentheses, such that the sum of the edit distance of these subsequences provides an $O(\log n)$ approximation of $\dyck(x)$. This reduction establishes a strong structural resemblance between string and Dyck edit distance though the dynamic program algorithms for these problems portrait a significant difference in the time complexity. Apart from this similar to string edit distance, Dyck edit distance also satisfies triangular inequality; precisely one can show for any three strings $x,y,z$, $\dyck(x\overline{z})\le \dyck(x \overline{y})+\dyck(y\overline{z})$, here $\overline{z}$ is the reverse complement of $z$ (obtained by reversing $z$ and flipping the direction of each parenthesis).
We remark that unlike string edit distance, establishing triangular inequality for Dyck edit distance is not straightforward and require several technical arguments.
Motivated by these similarities we try to adapt the three-step framework algorithm to approximate Dyck edit distance. However, it turns out that the very first step \emph{window decomposition} fails for our purpose. 
Thus, to design a subquadratic-time $O(1)$ approximation algorithm for Dyck edit distance we
\begin{enumerate}
    \item Propose a new window decomposition strategy and provide an analysis that shows that any optimal alignment of $x$ can be represented by matching the window pairs generated by our strategy.
    
    \item Establish triangular inequality for Dyck edit distance. 
\end{enumerate}

Next we discuss the limitations of the windowing strategy of~\cite{CDGKS18} and explain how to overcome them. 

~\cite{CDGKS18} starts by dividing each input strings into fixed sized substrings and then after finding the distance between different pair of substrings (such that the optimal alignment is covered) a dynamic program algorithm is used to combine them.
In our case the first difference is that the input to the Dyck edit distance problem is just a single string. Regardless, following the strategy of dimension reduction, one idea could be to partition the input string $x$ into windows $w_1,\dots,w_\ell \subset [1\dd n]$ of length $d$ (where $d=n^{\theta(1)}$) with the hope that provided the Dyck edit distance for each substring $x[w_i]\circ x[w_j]$ (here $x[w_i]$ denotes the substring of $x$ restricted to the indices in $w_i$ and $x[w_i]\circ x[w_j]$ represents concatenation of substrings $x[w_i]$ and $x[w_j]$), one can use the cubic time dynamic program algorithm to compute an estimation of $\dyck(x)$ (Dyck edit distance of $x$) in time $\tOh(\frac{n^3}{d^3})$. However, it turns out that unlike the string edit distance this straight forward decomposition does not work directly for the following reasons. 

\begin{itemize}
    \item In case of string edit distance if an optimal alignment (with cost $d$) matches $x[i]$ with $y[j]$ then $x[i+1]$ can be matched only with a character in $x[i+1],\dots, x[i+d+1]$.
    Thus, if we consider a window $w_1$ in $x$ and a window $w_2$ in $y$ such that $ED(x[w_1],x[w_2])$ is small then we can say $|w_1|\approx |w_2|$. This structural property completely breaks down for Dyck edit distance. For example in Fig~\ref{fig:fig1} for the two windows $w_1,w_2$ though their sizes are very different, $\dyck(x[w_1]\circ x[w_2])=0$. This shows partitioning $x$ into a single length windows does not suffice.
    \item 
    To overcome the above-mentioned issue let us assume that we allow variable length windows. However, it is possible that an optimal alignment matches a window $w_1$ of length $d$ with a window $w_2$ of length $\gg d$
   and thus if we allow windows to have length $\gg d$ (it can be as large as $\Omega(n)$) then evaluating an estimation of the cost of these large window pairs may not be efficient. %$\dyck(x[w_1]\circ x[w_2])=0$ using the quadratic time constant approximation algorithm (see Theorem~\ref{thm:approxdyck}) may not be efficient. 
One way out could be to further divide both $w_1$ and $w_2$ into smaller windows $w_1^1,w_1^2,\dots, w_1^k$ and $w_2^1,w_2^2,\dots, w_2^k$ respectively
and separately compute cost of each substring $x[w_1^i]\circ x[w_2^j]$.
Here as $|w_1^i|$ and $|w_2^j|$ are not too large, their distance can be computed efficiently but
notice as $|w_1^i|$ can be very small (as small as $n^{o(1)}$); in this case the total number of subproblems (total number of window pairs that we evaluate can grow) can explode and thus the dynamic program algorithm combining these subproblems may become inefficient.
\end{itemize}
Thus, for Dyck edit distance one of the most challenging task is to divide the input string $x$ into a set of variable sized windows such that the window sizes are neither too small (this ensures that the total number of subproblems (window pairs) we evaluate is not too large; thus the DP combining them is efficient) nor too large (so that evaluating the distance of a pair of windows is efficient as well) and any optimal alignment of $x$ can be represented by matching these window pairs.
Formally we need to construct a set of windows
$\mathcal{J}$ where each window has length at most $d$ (we set this $d$ to be a polynomial in $n$), $|\mathcal{J}|\approx \frac{n}{d}$ and there exists a subset $\mathcal{S}=\{(w_1,w'_1),\dots,(w_\ell,w'_\ell)\}\subseteq \mathcal{J}\times \mathcal{J}$ such that $\mathcal{S}$ is a consistent window decomposition of $[1\dd n]$ (i.e., no window pairs cross each other) and there exists an optimal alignment $M$ that matches $w_i$ with $w'_i$ (modulo some error ensuring constant approximation).
%for each optimal alignment $M$ there exists a subset $\mathcal{S}=\{(w_1,w'_1),\dots,(w_\ell,w'_\ell)\}\subseteq \mathcal{J}\times \mathcal{J}$ such that $\mathcal{S}$ is a consistent window decomposition of $[1,n]$ (see Definition~\ref{def:consistent}) and $M$ matches $w_i$ with $w'_i$. 
Constructing such a set $\mathcal{J}$ and showing it actually satisfies the above-mentioned property is one of the novelty of our algorithm and significantly differs from the window decomposition of~\cite{CDGKS18}. Next we discuss the construction and analyze some important properties of $\mathcal{J}$.

\paragraph{Window Decomposition.}
We construct the set of windows into two stages. 

\vspace{1mm}

\noindent
\textbf{Stage 1.} In stage 1 given a window size parameter $s_1$ and a distance threshold parameter $\theta$ (here, instead of directly estimating $\dyck(x)$, we distinguish between $\dyck(x)\le \theta n$ and $\dyck(x)>c \theta n$ for constant $c$), we construct a set of windows $\mathcal{J}$ where for each window $w\in \mathcal{J}$, $w\subseteq [1\dd n]$, $|w|\le 5s_1$, the last index of $w$ (denoted $e(w)$) and the length of $w$ are multiples of $\theta s_1$.
Notice though the window decomposition strategy is comparatively simple, it is nontrivial to show that any optimal alignment $M$ can be represented by matching window pairs from set $\mathcal{J}\times \mathcal{J}$. We remark that this analysis is particularly challenging and requires new combinatorial insights. Formally we show the following.

\begin{restatable}{lemma}{lemlarge}\label{lem:large}
	For every string of length $n$ such that $\theta s_1$ is a multiple of $n$, 
	the interval $[1\dd n]$ can be decomposed into a consistent set of window pairs $\mathcal{S}\sub \mathcal{J}\times \mathcal{J}$ such that $\sum_{(w,w')\in \mathcal{S}} \dyck(x[w]\circ x[w']) \le \dyck(x)+\Oh(\theta n)$.
\end{restatable}

%the interval $[1\dd n]$ can be decomposed into a consistent set of window pairs $\mathcal{S}\subseteq \mathcal{J}\times \mathcal{J}$ such that $\sum_{(w,w')\in \mathcal{S}}\dyck(x[w]\circ x[w'])$ provides a constant approximation of $\dyck(x)$.  

Next we provide a proof sketch of the Lemma~\ref{lem:large}. Note, as for every window in $w\in \mathcal{J}$, the length and the last index is a multiple of $\theta s_1$ (here $|w|\le 5s_1$) in order to upper bound $\sum_{(w,w')\in \mathcal{S}} \dyck(x[w]\circ x[w'])$ we need 
the size of set $\mathcal{S}$ to be $O(n/s_1)$. This suggests that for most of the window pairs, the sum of their lengths should be roughly $s_1$ (though unlike edit distance, here the windows can have varied lengths). 
We now briefly argue that such a set $\mathcal{S}$ indeed exists by providing a construction of it. Precisely we show for every optimal alignment we can construct a corresponding set $\mathcal{S}$ satisfying the above-mentioned properties.
We start by fixing an optimal alignment $M$ (Note we do not have any prior information of $M$. We consider one for the analysis purpose.). Given $M$, we construct $\mathcal{S}$ recursively where at each step we process an interval $[i_1\dd i_2]\subseteq [n]$ as follows. 
\begin{enumerate}
    \item If there exists a pair of windows $(w_1,w_2)$ around the left and the right boundary of $[i_1\dd i_2]$ (i.e., $w_1$ starts at $i_1$ and $w_2$ ends at $i_2$) such that $M$ matches $w_1$ with $w_2$ and $|w_1|+|w_2|$ is between $s_1$ and $4s_1$ we add pair $(w_1,w_2)$ to $\mathcal{S}$ and recurse on the rest of the interval $[i_1\dd i_2]\setminus (w_1\cup w_2)$. 
    
    \item Another favorable scenario is where there exists a pivot $p$ (here  $p\in [i_1\dd i_2)$ is a pivot satisfying $\dyck(x[i_1\dd i_2])=\dyck(x[i_1\dd p])+\dyck(x(p,i_2])$ and $p$ can be easily identified from $M$) roughly at distance $s_1$ either from the left or the right boundary of the interval $[i_1\dd i_2]$. In this case we can define a window pair $(w_1,w_2)$ covering the region from the boundary to the pivot and recurse on the rest. 
    
    \item In the last scenario let $(w_1,w_2)$ be the pair of windows around the left and the right boundary of $[i_1\dd i_2]$ such that $M$ matches $w_1$ with $w_2$ (notice here $|w_1|+|w_2|$ can be small) and let $p$ the leftmost pivot (note here the pivot is far i.e., at distance $>s_1$ from both left and right boundary), add $(w_1,w_2)$ to $\mathcal{S}$ and separately recurse on the intervals left and right of $p$ (excluding $w_1\cup w_2$).
\end{enumerate}

Here we can make two important observations. First, all the window pairs generated in this process are disjoint, and they provide a decomposition of $[1\dd n]$. Secondly we can ensure at least for the first two cases the sum of the length of the window pairs is around $s_1$. 
However, this is not true for the third case, and it becomes difficult to bound $|\mathcal{S}|$, as we may always land on the third case and produce window pairs of small lengths and thus  $|\mathcal{S}|$ increases. 
Fortunately we can show that this is unlikely and the third case cannot occur too often. 
This is because in this case the pivot $p$ is far (at distance roughly $>s_1$) from both the left and right boundary and therefore the current interval is always divided into two large intervals (where each interval has size at least roughly $s_1$) and thus the recursion depth should be small.
Following this, we can show that total number of windows in $\mathcal{S}$ is $O(n/s_1)$. This is important as in the next step we cap each window in $\mathcal{S}$ by shifting both the starting and the end index to the right nearest multiple of $\theta s_1$ to ensure $\mathcal{S}\subseteq \mathcal{J}\times \mathcal{J}$. Notice as this adds $\theta s_1$ extra cost per pair compared to the optimal one, using the bound on size of $\mathcal{S}$, we can ensure constant approximation. To summarize, we show for each window pair in $\mathcal{J}\times \mathcal{J}$ if we can compute an estimation of its cost then this trivially provides a cost estimation for all window pairs in $\mathcal{S}$. Thus, as $\mathcal{S}$ partitions $x$, using a dynamic program algorithm on window pairs from $\mathcal{J}\times \mathcal{J}$, we can optimize the total cost.

\vspace{1mm}

\noindent
\textbf{Stage 2.}
To compute the cost of each window pair in $\mathcal{J}\times \mathcal{J}$
we further partition each large window into smaller windows and estimate the cost of these smaller windows.
Thus, in stage 2 given a second window size parameter $s_2$, we further divide interval $[1\dd n]$ into variable sized shorter windows $\mathcal{K}$, where the length of a window is a multiple of $\theta s_2$. Following a similar argument as above we can show for each window pair $(w,w')\in \mathcal{J}\times \mathcal{J}$ the set $w\cup w'$ can be decomposed into a consistent set of window pairs $\mathcal{S}_{(w,w')}\subseteq \mathcal{K}\times \mathcal{K}$ such that $\sum_{(q,q')\in \mathcal{S}_{(w,w')}}\dyck(x[q]\circ x[q'])$ provides a constant approximation of $\dyck(x[w]\circ x[w'])$ (\cref{lem:small} is formulated analogously to \cref{lem:large}).

\paragraph{CertifyingWindowPairs.} Next
in step two and three we design algorithm CertifyingWindowPairs() that finds a cost estimation for each window pair $(w,w')\in \mathcal{J}\times \mathcal{J}$.
This algorithm has two components; CertifyingSmall() and CertifyingLarge(). 
CertifyingWindowPairs() shares a similar flavor with the Covering algorithm of~\cite{CDGKS18}. 
Here also to get a sub-quadratic time bound, instead of evaluating the cost of each window pair explicitly, we need to use triangle inequality.
%Even though the Dyck edit distance is a single-argument function, it naturally yields a similarity measure defined for two strings $x,y$ as $\dyck(x\circ \rev{y})$, where $\rev{y}$ is the reverse complement of $y$ (obtained by reversing $y$ and flipping the direction of each parenthesis).
In \cref{lem:triangle}, we show that Dyck edit distance satisfies the triangle inequality; but unlike for the regular edit distance, this property of Dyck edit distance is non-trivial, and we prove it using a subtle inductive argument. 
%Next for completeness we provide a brief description of 
We provide the details of algorithms CertifyingSmall() and CertifyingLarge() in Section~\ref{sec:cost_est}.

\paragraph*{Folding Distance.}
The key difference between the \emph{folding distance} (originating from the RNA folding problem) compared to the Dyck edit distance is that the alphabet is no longer partitioned into the set $T$ of opening parentheses and the set $\rev{T}$ of closing parentheses. In other words, every character $c$ can be matched with its complement $\rev{c}$ regardless of their order in the text.
In particular, this means that the notions of heights and valleys are not meaningful anymore.
Nevertheless, for any fixed alignment, one can distinguish the unmatched,
opening (matched with a character to the right), and closing (matched with a character to the left) characters.
Moreover, with a reduction similar to that of~\cite{BO16}, we can make sure that there is an unmatched character between any two characters matched with each other.
Although this reduction does not seem helpful in sparsifying the set of pivots to be considered,
it does bring a strong structural property: for every cost-$d$ alignment,
the matched characters form $\Oh(d)$ disjoint windows such the alignment matches (pairs) entire windows without any errors.
The strategy behind our $\Oh(s)$-factor approximation is to sacrifice $\Oh(s)$ boundary characters out of each window pair and, in exchange, make sure that the `closing windows' (the right ones out of every pair matched windows) have both endpoints at positions divisible by $s$. 
We then use \textsc{Internal Pattern Matching} queries of~\cite{KRRW15,phd} to efficiently search for `opening windows' that could match our closing windows. Doing so, we cannot guarantee that the opening windows have their endpoints divisible by $s$, but we can sparsify the set of candidates so that they start at least $s$ positions apart. This results in $\Oh((\frac{n}{s})^3)$ window pairs to be considered and leads to the overall running time of $\Oh(n+(\frac{n}{s})^3)$ [Theorem~\ref{thm:fold}].
%{\color{red} Doesn't the following paragraph need to be changed now?}

\section{Preliminaries}\label{sec:prelim}
The alphabet $\Sigma$ consists of two disjoint sets $T$ and $\rev{T}$ of \emph{opening} and \emph{closing} parentheses, respectively, with a bijection $\rev{\cdot} : T\to \rev{T}$ mapping each opening parenthesis to the corresponding closing parenthesis. We extend this mapping to an involution $\rev{\cdot}:T\cup \rev{T}\to T\cup\rev{T}$ and,
in general, to an involution $\rev{\cdot} : \Sigma^*\to \Sigma^*$ mapping each string 
$x[1] x[2]\cdots x[n]$ to its complement $\rev{x[n]}\cdots \rev{x[2]}\,\rev{x[1]}$. 
Given two strings $x,y$, we denote their concatenation by $xy$ or $x\circ y$.

The \emph{Dyck} language $\DYCK(\Sigma)\sub \Sigma^*$ consists of all well-parenthesized expression over $\Sigma$;
formally, it can be defined using a context-free grammar whose only non-terminal $S$ admits productions $S\rightarrow SS$, $S\rightarrow \varnothing$ (empty string), and $S\rightarrow aS\overline{a}$ for all $a\in T$. 

\begin{definition}\label{def:dyckdist}
The \emph{Dyck edit distance} $\dyck(x)$ of a string $x\in \Sigma^*$
is the minimum number of character insertions, deletions, and substitutions required to transform $x$ to a string in $\DYCK(\Sigma)$.
\end{definition}

\newcommand{\Z}{\mathbb{Z}}

We say that $M\sub \{(i,j) \sub \Z^2 : i < j\}$ is a \emph{non-crossing matching} if any two distinct pairs $(i,j),(i',j')\in M$ satisfy $i<j < i' < j'$ or $i < i' < j' < j$. Such a matching can also be interpreted as a function $M : \Z \to \Z\cup\{\bot\}$ with $M(i)=j$ if $(i,j)\in M$ or $(j,i)\in M$ for some $j\in \Z$, and $M(i)=\bot$ otherwise.

For a string $x\in \Sigma^n$, the \emph{cost} of a non-crossing matching $M\sub [n]^2$ on $x$ (henceforth $M$ is called an \emph{alignment} of~$x$) is defined as
\[\cost_M(x) =n-2|M| + \sum_{(i,j)\in M}\dyck(x[i]x[j]).\]

The following folklore fact, proved below for completeness, relates the Dyck edit distance with the optimum alignment cost.
\begin{restatable}{fact}{fctal}
For every string $x\in \Sigma^*$, the Dyck edit distance $\dyck(x)$ is the minimum cost $\cost_M(x)$ of an alignment $M$ of $x$.
\end{restatable}
\begin{proof}
	We first show that $\dyck(x)\le \cost_M(x)$ by induction on $|x|$.
	The claim is trivial if $|x|=0$.
	If $M(1)=\bot$, then we construct $M'=\{(i-1,j-1) : (i,j)\in M\}$ and $x'=x[2\dd |x|]$.
	By the inductive assumption, $\dyck(x)\le \dyck(x')+1\le \cost_{M'}(x') = \cost_M(x)$.
	If $M(1)=|x|$, then we construct $M'=\{(i-1,j-1) : (i,j)\in M\sm \{(1,|x|)\}$ and $x'=x[2\dd |x|)$.
	By the inductive assumption, $\dyck(x)\le \dyck(x')+\dyck(x[1]x[|x|])\le \cost_{M'}(x') + \dyck(x[1]x[|x|]) = \cost_M(x)$.
	Otherwise, we have $(1,p)\in M$ for some $p\in [2\dd |x|)$.
	In this case, we construct $M'=\{(i,j)\in M : j\le p\}$ and $x' = x[1\dd p]$,
	as well as $M'' = \{(i-p,j-p) : (i,j)\in M\text{ and }i>p\}$ and $x''=x[p+1\dd |x|]$.
	By the inductive assumption, $\dyck(x)\le \dyck(x')+\dyck(x'')\le \cost_{M'}(x') + \cost_{M''}(x'') = \cost_M(x)$;
	here, the last equality follows from the fact that $|M|=|M'|+|M''|$: any $(i,j)\in M$ with $i \le p$ and $j>p$ would violate the non-crossing property of $M$.
	
	As for the inverse inequality, we proceed by induction on $2\dyck(x)+|x|$;
	again, the claim is trivial for $|x|=0$.
	If $x\in \DYCK(\Sigma)$ and $x=ax'\bar{a}$ for $a\in T$ and $x'\in \DYCK(\Sigma)$, then the inductive assumption yields an alignment $M'$ with $\cost_{M'}(x')=0$. In this case, we set $M=\{(i+1,j) : (i,j)\in M'\}\cup\{(1,|x|)\}$, so that $\cost_M(x)=\cost_{M'}(x')=0$.
	If $x = x'x''$ for non-empty $x',x''\in \DYCK(\Sigma)$, then the inductive assumption yields alignments
	$M',M''$ with $\cost_{M'}(x')=\cost_{M''}(x'')=0$. In this case, we set $M=M'\cup\{(i+|x'|,j+|x'|) : (i,j)\in M''\}$,
	so that $\cost_M(x)=\cost_{M'}(x')+\cost_{M''}(x'')=0$.
	We henceforth assume that $\dyck(x)>0$ and consider the first operation in the sequence of operations transforming $x$ to a string in $\DYCK(\Sigma)$. The inductive hypothesis on the resulting string $x'$ yields an alignment $M'$ with $\cost_{M'}(x')=\dyck(x')=\dyck(x)-1$.
	If $x'$ is obtained from $x$ by a substitution of $x[p]$, then we set $M'=M$ noting that $\cost_M(x)\le \cost_{M'}(x')+1 = \dyck(x')+1=\dyck(x)$: the position $p$ is involved in at most one pair $(i,j)\in M'$ and this pair satisfies $\dyck(x[i]x[j])\le \dyck(x'[i]x'[j])+1$.
	If $x'$ is obtained from $x$ by deleting $x[p]$, then we set $M = \{(i,j)\in M' : j < p\} \cup \{(i,j+1) : (i,j)\in M' \text{ and } i< p \le j\}\cup\{(i+1,j+1): (i,j)\in M'\text{ and }p \le i\}$,
	noting that $\cost_M(x)=\cost_{M'}(x')+1=\dyck(x')+1=\dyck(x)$.
	Finally, if $x'$ is obtained from $x$ by inserting $x'[p]$, then $M = \{(i,j)\in M' : j < p\} \cup \{(i,j-1) : (i,j)\in M' \text{ and } i< p < j\}\cup\{(i-1,j-1): (i,j)\in M'\text{ and }p < i\}$,
	noting that $\cost_M(x)\le \cost_{M'}(x')+1=\dyck(x')+1=\dyck(x)$: if $M'$ pairs $p$ with some position,
	then the corresponding position of $x$ contributes to $|x|-2|M|$.
	\end{proof}

Next, we show that a function mapping $x,y\in \Sigma^*$ to $\dyck(x\rev{y})$ 
satisfies the triangle inequality.

\begin{lemma}\label{lem:triangle}
	All strings $x,y,z\in \Sigma^*$ satisfy $\dyck(x\rev{z}) \le \dyck(x\rev{y}y\rev{z})\le \dyck(x\rev{y})+\dyck(y\rev{z})$. 
\end{lemma}
\begin{proof}
	The second inequality follows from the fact that the Dyck language is closed under concatenations.
	As for the first inequality, we observe that it suffices to consider $|y|=1$: the case of $|y|=0$ is trivial,
	and the case of $|y|>1$ can be derived from that of $|y|=1$ by processing $y$ letter by letter.
	Now, we proceed by induction on $2\dyck(x\rev{y}y\rev{z})+|x|+|z|$.
	If any optimum alignment of $x\rev{y}y\rev{z}$ modifies a character in $x$ or $\rev{z}$,
	we apply the inductive assumption for an instance $(x',y,z')$ obtained from this modification: $\dyck(x\rev{z}) \le \dyck(x'\rev{z}')+1 \le \dyck(x'\rev{y}y\rev{z}')+1 = \dyck(x\rev{y}y\rev{z})$.
	If any optimum alignment of  $x\rev{y}y\rev{z}$ matches any two adjacent characters of $x$,
	any two adjacent characters of $z$, or the first character of $x$ with the last character of $z$,
	we apply the inductive assumption for an instance $(x',y,z')$ obtained by removing these two characters:
	$\dyck(x\rev{z}) \le \dyck(x'\rev{z}') \le \dyck(x'\rev{y}y\rev{z}') = \dyck(x\rev{y}y\rev{z})$.
	In the remaining case, all characters of $x$ and $\rev{z}$ must be matched to $\rev{y}$ or $y$,
	so $|x\rev{z}|\le 2$. 
	If $|x\rev{z}|\le \dyck(x\rev{y}y\rev{z})$, then trivially $\dyck(x\rev{z})\le |x\rev{z}|\le \dyck(x\rev{y}y\rev{z})$,
	so we may assume $\dyck(x\rev{y}y\rev{z})<|x\rev{z}|$.
	The case of $\dyck(x\rev{y}y\rev{z})=0$ and $|x\rev{z}|=1$ is impossible because only strings of even length belong to the Dyck language.
	Thus, we may assume that $|x\rev{z}|=2$ and $\dyck(x\rev{y}y\rev{z})\le 1$.
	If $|x|=2$, then the optimum matching of $x\rev{y}y\rev{z}$ must be $\{(1,4),(2,3)\}$,
	and the sequence transforming $\dyck(x\rev{y}y\rev{z})$ to a word in $\DYCK(\Sigma)$
	must include substituting $\rev{y}$ or $y$ (whichever is an opening parenthesis).
	In particular, $x[1]$ must be an opening parenthesis, so $\dyck(x\rev{z})=\dyck(x)\le 1 = \dyck(x\rev{y}y\rev{z})$.
	If $|\rev{z}|=2$, then the optimum matching of $x\rev{y}y\rev{z}$ must be $\{(1,4),(2,3)\}$,
	and the sequence transforming $\dyck(x\rev{y}y\rev{z})$ to a word in $\DYCK(\Sigma)$
	must include substituting $\rev{y}$ or $y$ (whichever is a closing parenthesis).
	In particular, $z[1]$ must be an opening parenthesis, so $\dyck(x\rev{z})=\dyck(\rev{z})\le 1 = \dyck(x\rev{y}y\rev{z})$.
	Finally, if $|x|=|\rev{z}|=1$, then the optimum matching of $x\rev{y}y\rev{z}$ must be $\{(1,2),(3,4)\}$.
	If $\dyck(x\rev{y}y\rev{z})=0$, then we must have $x=y=z\in T$, so $\dyck(x\rev{z})=0\le \dyck(x\rev{y}y\rev{z})$.
	Otherwise, $x\in T$ or $z\in T$, so $\dyck(x\rev{z})\le 1 = \dyck(x\rev{y}y\rev{z})$.
\end{proof}

In the remainder of this section, we recall several results from~\cite{BO16,otherSubmission} that we then use in our $\tOh_{\epsilon}(n^2)$-time PTAS (\cref{sec:PTAS}) and $\tOh_{\epsilon}(nd)$-time $(3+\epsilon)$-approximation (\cref{sec:3apx}).

\begin{figure}
\begin{center}
\begin{tikzpicture}[scale=0.6]

	\draw (0,0)-- node[shift={(-0.2,0.1)}]{$\texttt{(}$} (1,1) -- node[shift={(-0.2,0.1)}]{$\texttt{[}$} (2,2) -- node[shift={(0.2,0.1)}]{$\texttt{)}$} (3,1) --node[shift={(-0.2,0.1)}]{$\texttt{[}$} (4,2) --node[shift={(-0.2,0.1)}]{$\texttt{(}$} (5,3) --node[shift={(0.2,0.1)}]{$\texttt{]}$} (6,2) -- node[shift={(0.2,0.1)}]{$\texttt{]}$} (7,1) -- node[shift={(-0.2,0.1)}]{$\texttt{(}$} (8,2) -- node[shift={(0.2,0.1)}]{$\texttt{]}$} (9,1) -- node[shift={(0.2,0.1)}]{$\texttt{)}$} (10,0) -- node[shift={(0.2,0.1)}]{$\texttt{)}$} (11,-1);

	\draw[help lines, color=gray!30, dashed] (0,-1) grid (11,3);
\draw[->,thick] (0,0)--(11.25,0) node[right]{$ $};
\draw[->,thick] (0,-1)--(0,3.25) node[shift={(0.2,0.1)}]{$ $};

\draw[blue, dotted, very thick, bend right=10] (0.5,0.5) to (10.5,-0.5);
\draw[blue, dotted, very thick, bend right] (1.5,1.5) to (8.5,1.5);
\draw[blue, dotted, very thick, bend right] (3.5,1.5) to (6.5,1.5);
\draw[blue, dotted, very thick, bend right] (4.5,2.5) to (5.5,2.5);

\filldraw (3,1) circle (0.1);
\filldraw (7,1) circle (0.1);
\draw[fill=white] (2,2) circle (0.1);
\draw[fill=white] (4,2) circle (0.1);
\draw[fill=white] (6,2) circle (0.1);
\draw[fill=white] (8,2) circle (0.1);
\end{tikzpicture}
\end{center}
\caption{A plot of the height function $h$ for $x=\texttt{([)[(]](]))}$. The blue dotted lines represent an alignment $M=\{(1,11),(2,9),(4,7),(5,6)\}$ of cost $4$. The valleys $\{3,7\}$ are marked as black circles. The set $K=\{2,3,4,6,7,8\}$ of \cref{fct:dp} also includes points marked as white circles.}\label{fig:h}
\end{figure}
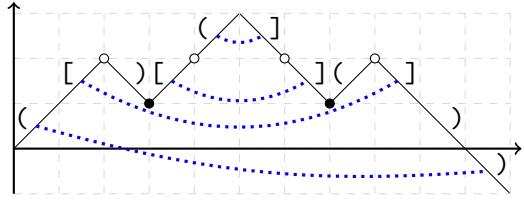

\begin{definition}[Heights]
For a fixed string $x\in \Sigma^n$, the \emph{height} function $h:[0\dd n]\to [-n\dd n]$
is defined so that $h(i)=|\{j\in [1\dd i] : x[j]\in T\}| - |\{j\in [1\dd i] : x[j]\in \rev{T}\}|$ for $i\in [0\dd n]$.
\end{definition}

\begin{fact}[\cite{BO16}]\label{fct:preprocess}
	There is a linear-time algorithm that, given a string $x\in \Sigma^n$,
	produces a string $x'\in \Sigma^{\le n}$ such that $\dyck(x)=\dyck(x')$
	and $x'$ has at most $2\dyck(x)$ \emph{valleys}, i.e., positions $v\in [1\dd n)$ such that $h(v-1)>h(v)<h(v+1)$.
\end{fact}

\SetKwFunction{D}{D}
For a fixed string $x\in \Sigma^n$, let us define a function $\D$ such that $\D(i,j)=\dyck(x(i\dd j])$
for $i,j\in [0\dd n]$ with $i\le j$.
Note that $\D(i,i)=0$ for $i\in [0\dd n]$, $\D(i,i+1)=1$ for $i\in [0\dd n)$,
and $\D(i,j)$ satisfies the following recursion for $i,j\in [0\dd n]$ with $j-i\ge 2$:
\begin{equation}\label{eq:basic_recursion}\D(i,j) = \min \begin{cases}
	\D(i,k) + \D(k,j) \quad \text{for } k\in (i\dd j),\\
	\D(i+1,j-1)+\dyck(x[i+1]x[j]).
\end{cases}\end{equation}
This yields the classic $\Oh(n^3)$-time algorithm computing $\D(0,n)=\dyck(x)$.
The following result, combined with \cref{fct:preprocess}, improves this time complexity to $\Oh(n + n^2 \dyck(x))$.

\begin{restatable}[{\cite[Lemma 2.1]{otherSubmission}}]{fact}{fctdp}\label{fct:dp}
	For a string $x\in \Sigma^n$, let $K\sub [0\dd n]$ consist of all positions at distance 0 or 1 from valleys. 
	For all $i,j\in [0\dd n]$ with $j-i\ge 2$, we have
	\begin{equation}\label{eq:faster_recursion}
		\D(i,j)= \min\begin{cases}\D(i,k) + \D(k,j) \quad \text{for } k\in (i\dd j)\cap (K\cup \{i+1,i+2,j-1,j-2\}),\\
		\D(i+1,j-1)+\dyck(x[i+1]x[j]).\end{cases}\end{equation}
\end{restatable}

\begin{observation}[{\cite[Fact 3.1]{otherSubmission}}]\label{obs:hd}
	For all strings $x\in \Sigma^n$ and integers $0\le i \le k \le j \le n$,
	we have $h(k) \ge \max(h(i),h(j))-2\D(i,j)$. In particular, $|h(i)-h(j)|\le 2\D(i,j)$.
\end{observation}

\section{Quadratic-Time PTAS}\label{sec:PTAS}
\SetKwFunction{AD}{AD}

\begin{algorithm}[b!]
\SetKwBlock{Begin}{}{end}
$\AD(i,j)$\Begin{
	\lIf{$j=i$}{\Return{$0$}}
	\lIf{$j=i+1$}{\Return{$1$}}
	$c := \AD(i+1,j-1)+\dyck(x[i+1]x[j])$\;
	$\tau_{i,j} := \tau\cdot {2^{\min(\nu(i),\nu(j))}}$\;
	$K_{i,j} :=$ the set of $\tau_{i,j}$ smallest and $\tau_{i,j}$ largest elements of $K\cap (i\dd j)$\;
	\ForEach{$k\in K_{i,j}\cup(\{i+1,i+2,j-2,j-1\}\sm \{i,j\})$}{
		$c := \min(c, \AD(i,k)+\AD(k,j))$\;
	}
	\Return{$c$}\;
	}
\caption{Recursive implementation of $\protect\AD(i,j)$}\label{alg:Dyck-Approx}
\end{algorithm}

In this section, we develop an $\tOh(\epsilon^{-1}n^2)$-time algorithm that approximates $\dyck(x)$ within a $(1+\epsilon)$ factor. The starting point of our solution is the dynamic program derived from \cref{fct:dp}.
Instead of computing the exact value $\D(i,j)=\dyck(x(i\dd j])$, that depends on $\D(i,k)+\D(k,j)$ for all pivots $k\in  (i\dd j)\cap (K\cup \{i+1,i+2,j-1,j-2\})$, we compute an approximation $\AD(i,j)\approx \dyck(x(i\dd j])$ in \cref{alg:Dyck-Approx} that depends only on $\D(i,k)+\D(k,j)$ for pivots  $k\in  (i\dd j)\cap (K_{i,j}\cup \{i+1,i+2,j-1,j-2\})$, where $K_{i,j}$ consists of $\tau_{i,j}$ leftmost and rightmost elements of $K\cap (i\dd j)$. Here, $\tau_{i,j}$ is proportional to the largest power of two dividing both $i$ and $j$.
Formally, we set $\tau_{i,j} := \tau\cdot {2^{\min(\nu(i),\nu(j))}}$, where $\tau$ is a parameter to be set later
and $\nu : \Z \to \Z_{\ge 0}\cup\{\infty\}$ is a function
that maps an integer $r\in \Z$ to $\nu(r):= \max\{k\in \Z : 2^k \text{ divides } r\}$, with the convention that $\nu(0)=\infty$.

In the following lemma, we inductively bound the quality of $\AD(i,j)$ as an additive approximation  of $\D(i,j)$.
In particular, we show that $\D(i,j) \le \AD(i,j)\le \D(i,j) + \frac{8}{\tau} |K|\log |K|$.
\begin{lemma}\label{lem:AD}
	If $\tau \ge 2$, then, for each $i,j\in [0\dd n]$ with $i\le j$, we have $\D(i,j)\le \AD(i,j)\le \D(i,j) + \frac{8}{\tau} c_{i,j}\log c_{i,j}$,
	where $c_{i,j} :=|K\cap (i\dd j)|$, and we assume $0\log 0 = 0$.
\end{lemma}
\begin{proof}
	We proceed by induction on $j-i$. 	For $j-i\le 1$, we have $\AD(i,j)=\D(i,j)$.
	For $j-i \ge 2$, the lower bound $\D(i,j)\le \AD(i,j)$ follows directly from \cref{fct:dp}.
	Unless $\D(i,j)=\D(i,k)+\D(k,j)$ for some $k\in (i\dd j)\cap K$,
	the upper bound also follows by \cref{fct:dp}.
	Let $r=\min(c_{i,k},c_{k,j})$ and let $i',j'$ be the smallest and the largest multiples of $2^{\lceil\log ((r+1)/\tau)\rceil}$ within $[i\dd j]$. We claim that $k\in K_{i',j'}$.

	Note that $\tau(i'-i) < \tau 2^{\lceil\log ((r+1)/\tau)\rceil} < 2(r+1)$, so $\tau(i'-i)\le 2r$;
	symmetrically,  $\tau(j-j')\le 2r$.
	Due to $\tau\ge 2$, we thus have $i'-i\le r \le c_{i,k}< k-i$ and $j-j'\le r \le c_{j,k} < j-k$,
	so $k\in (i'\dd j')$. Moreover, $\tau_{i',j'}\ge \tau\cdot 2^{\lceil\log ((r+1)/\tau)\rceil}
	\ge r+1 \ge \min(c_{i',k},c_{k,j'})+1$, so $k\in K_{i',j'}$ holds as claimed.
	Thus, \begin{align*}
		\AD(i,j)&\le (i'-i)+ \AD(i',j')+(j-j')\\
		 & \le  \tfrac{2r}{\tau}+\AD(i',k)+\AD(k,j') + \tfrac{2r}{\tau}\\
		 &\le \D(i',k) + \tfrac{8}{\tau} c_{i',k}\log c_{i',k} + \D(k,j') + \tfrac{8}{\tau} c_{k,j'}\log c_{k,j'} +\tfrac{4r}{\tau} \\
		 &\le (i'-i)+\D(i,k) + \tfrac{8}{\tau} c_{i,k}\log c_{i,k} + \D(k,j)+ (j-j')+ \tfrac{8}{\tau} c_{k,j}\log c_{k,j} +\tfrac{4r}{\tau}	\\
		 &\le  \D(i,j) +\tfrac{8}{\tau} (c_{i,k}\log c_{i,k} +  c_{k,j}\log c_{k,j}+r)\\
		 &=  \D(i,j) +\tfrac{8}{\tau} (\max(c_{i,k},c_{k,j})\log \max(c_{i,k},c_{k,j}) +  \min(c_{i,k},c_{k,j})\log (2\min(c_{i,k},c_{k,j})))\\
		 &\le \D(i,j) + \tfrac{8}{\tau} (\max(c_{i,k},c_{k,j})\log c_{i,j} + \min(c_{i,k},c_{k,j}) \log c_{i,j}) \\
		 &\le  \D(i,j) + \tfrac{8}{\tau} c_{i,j}\log c_{i,j}.\qedhere\end{align*}
\end{proof}

Our final solution simply uses \cref{alg:Dyck-Approx} with an appropriate choice of the parameter $\tau$
and the input string preprocessed using \cref{fct:preprocess} so that $|K|=\Oh(\dyck(x))$.
\begin{theorem}\label{thm:approxdyck}
	There is an algorithm Dyck-Approx that, given a string $x\in \Sigma^n$ and a parameter $\epsilon\in (0,1)$,
	in $\tOh(\epsilon^{-1}n^2)$ time computes a value $v$ such that $\dyck(x)\le v \le (1+\epsilon)\dyck(x)$.
\end{theorem}
\begin{proof}
In the preprocessing, we use \cref{fct:preprocess} in order to guarantee that there are at most $2\dyck(x)$ valleys and thus $|K|\le 6\dyck(x)$. 
Next, we call $\AD(0,n)$ with $\tau = \lceil 48 \epsilon^{-1} \log |K|\rceil$
and an array of size $(n+1)\times (n+1)$ memorizing the outputs of recursive calls.
The resulting value satisfies $\dyck(x)\le \AD(0,n) \le \dyck(x)+\frac{8}{\tau} |K|\log |K| \le (1+\epsilon)\dyck(x)$ by \cref{lem:AD}.
The running time is proportional to \begin{multline*}
	n^2 \sum_{i=0}^{n}\sum_{j=i+2}^n \tau_{i,j} \le n^2 + \sum_{i=0}^{n}\sum_{j=i+2}^n {\tau 2^{\nu(j)}}
	\le n^2 + n\tau \sum_{j=2}^n 2^{\nu(j)}\le n^2 + n\tau \sum_{\nu = 0}^{\lfloor{\log n}\rfloor}\left\lfloor{\frac{n}{2^\nu}}\right\rfloor 2^\nu =
	\\ = \Oh(n^2 \tau \log n)=\tOh(\epsilon^{-1} n^2).\qedhere\end{multline*}
\end{proof}

\section{Constant-Factor Approximation for Small Distances}\label{sec:3apx}
\SetKwFunction{GD}{GD}
\SetKwFunction{BD}{AGD}

In this section, we speed up the algorithm of \cref{sec:PTAS}
at the cost of increasing the approximation ratio from $1+\epsilon$ to $3+\epsilon$.
The key idea behind our solution is to re-use the DP of \cref{fct:dp} and \cref{alg:Dyck-Approx}
with an extra constraint that the transition from $(i,j)$ to $(i+1,j-1)$ (which corresponds to adding $(i+1,j)$ to the alignment $M$, i.e., matching $x[i+1]$ with $x[j]$) is forbidden if there is a deep valley within $(i\dd j)$.
This condition is expressed in terms of the following function:
\begin{definition}
For a fixed string $x\in \Sigma^n$ and $i,j\in [0\dd n]$ with $i\le j$, let $h(i,j)=\min_{k=i}^j h(k)$.
\end{definition}

Namely, we require that $h(i+1,j-1) > h(i,j)$ holds for all $(i+1,j)\in M$. For example, in the alignment $M$ of \cref{fig:h}, $(2,9)\in M$ violates this condition due to $h(1,9)=1 = h(2,8)$, whereas the remaining pairs satisfy this condition.
Formally, we transform the recursion of \cref{fct:dp} into the following one,
specified through a function  $\GD:[0\dd n]^2\to [0\dd n]$ such that 
 $\GD(i,i)=0$ for $i\in [0\dd n]$, $\GD(i,i+1)=1$ for $i\in [0\dd n)$,
and, for all $i,j\in [0\dd n]$ with $j-i\ge 2$:
\[\GD(i,j) = \min \begin{cases}
	\GD(i,k) + \GD(k,j) \quad \text{for } k\in (i\dd j)\cap (K\cup \{i+1,i+2,j-2,j-1\}),\\
	\GD(i+1,j-1)+\dyck(x[i+1]x[j]) \quad \text{if }h(i+1,j-1)>h(i,j).
\end{cases}\]

Somewhat surprisingly, this significant limitation on the allowed alignments $M$ incurs no more that a factor-$3$ loss in optimum alignment cost. Specifically, if we take an arbitrary alignment $M$ of $x$
and remove all pairs $(i+1,j)$ with $h(i+1,j-1) = h(i,j)$, the resulting alignment $M'$ satisfies
$\cost_{M'}(x)\le 3\cost_M(x)$. This can be proved by induction on the structure of $M$ 
using a potential function  $h(i)+h(j)-2h(i,j)$ as a ``budget'' for future deletions of matched  pairs.
Nevertheless, the proof of the following lemma operates directly on $\D$ and $\GD$.
\begin{restatable}{lemma}{lemgood}\label{lem:good}
Let $x\in \Sigma^n$. For all $i,j\in [0\dd n]$ with $i\le j$, we have 
$\D(i,j)\le \GD(i,j) \le 3 \D(i,j)-h(i)-h(j)+2h(i,j)$.
\end{restatable}
\begin{proof}
	We proceed by induction on $j-i$. The lower bound holds trivially.
	As for the upper bound, we consider several cases:
   \begin{itemize}
	   \item $j=i$. In this case, $\GD(i,j)= 0 = 3\cdot 0 - h(i)-h(i)+2h(i) = \D(i,j)-h(i)-h(j)+2h(i,j)$.
	   \item $j=i+1$. In this case, $\GD(i,j) = 1 < 2 =3\cdot 1 - h(i)-h(i+1)+2\min(h(i),h(i+1))=\D(i,j)-h(i)-h(j)+2h(i,j)$.
	   \item $\D(i,j)=\D(i,k)+\D(k,j)$ for some $k\in (i\dd j)\cap (K\cup \{i+1,i+2,j-2,j-1\})$.
	   In this case, we have \begin{align*}\GD(i,j) 
		&\le \GD(i,k)+\GD(k,j) \\
	   & \le 3\D(i,k)-h(i)-h(k)+2h(i,k)+3\D(k,j)-h(k)-h(j)+2h(k,j)\\
	   &= 3\D(i,j)-h(i)-h(j)-2h(k)+2\min(h(i,k),h(k,j))+2\max(h(i,k),h(k,j)) \\
	   &\le3\D(i,j)-h(i)-h(j)-2h(k)+2h(i,j) + 2h(k)  \\
	   & = 3\D(i,j)-h(i)-h(j)+2h(i,j).
	\end{align*}
	   \item $\D(i,j)=\D(i+1,j-1)+\dyck(x[i+1]x[j])$ and $h(i+1,j-1)=h(i,j)+1$.
	   In this case, we have \begin{align*}\GD(i,j)
	   &\le \GD(i+1,j-1) + \dyck(x[i+1]x[j])\\
	   &\le 3\D(i+1,j-1) - h(i+1)-h(j-1) + 2h(i+1,j-1)+ \dyck(x[i+1]x[j])\\
	   &= 3\D(i,j)- h(i+1)-h(j-1)+2h(i,j) - 2\dyck(x[i+1]x[j]) + 2\\
	   &\le 3\D(i,j) - h(i)-h(j)+2h(i,j)\end{align*}
	   because $2\dyck(x[i+1]x[j]) \ge 2+h(i)-h(i+1)+h(j)-h(j-1)$.
	   \item $\D(i,j)=\D(i+1,j-1)+\dyck(x[i+1]x[j])$ and $h(i+1,j-1)=h(i,j)$.
	   In this case, we have \begin{align*}\GD(i,j) 
	   &\le \GD(i,i+1)+\GD(i+1,j-1)+\GD(j-1,j) \\
	   & = \GD(i+1,j-1)+2\\
	   &\le 3\D(i+1,j-1) - h(i+1)-h(j-1) + 2h(i+1,j-1)+ 2\\
	   &= 3\D(i,j)- h(i+1)-h(j+1)+2h(i,j) - 3\dyck(x[i+1]x[j])+ 2\\
	   &\le 3\D(i,j) - h(i)-h(j)+2h(i,j)\end{align*} because $3\dyck(x[i+1]x[j]) \ge 2\dyck(x[i+1]x[j]) \ge 2+h(i)-h(i+1)+h(j)-h(j-1)$.\qedhere
   \end{itemize}
\end{proof}

Next, we derive a property of $\GD$ that allows for a speedup compared to $\D$.
Recall that $\GD$ forbids the transition from $(i,j)$ to $(i+1,j-1)$ if $h(i+1,j-1) = h(i,j)$.
We further show that, in this case, 
it suffices to consider one specific pivot while computing $\GD(i,j)$ (specifically, the pivot of minimum height, with ties resolved arbitrarily; see Fact~\ref{fct:gdp}). Later on (in the proof of \cref{thm:3approxdyck}), we argue that, after the preprocessing of \cref{fct:preprocess}, there are only $\Oh(nd)$ pairs $(i,j)$ for which $h(i+1,j-1) > h(i,j)$ yet $\D(i,j)\le d$.
%\tknote{I changed $=$ to $>$; the version with $=$ is not true.}

\begin{fact}\label{fct:gdp}
Let $x\in \Sigma^n$ and let $i,j\in [0\dd n]$ with $j-i\ge 2$ and $h(i+1,j-1) = h(i,j)$.
Then, every $k^*\in (i\dd j)$ with $h(k^*)=h(i,j)$ satisfies $\GD(i,j)=\GD(i,k^*)+\GD(k^*,j)$.
\end{fact}
\begin{proof}
We proceed by induction on $j-i$.
Fix $k\in (i\dd j)\cap (K\cup \{i+1,i+2,j-2,j-1\})$ such that $\GD(i,j)=\GD(i,k)+\GD(k,j)$.
If $k=k^*$, then the claim is trivial.
Thus, by symmetry, we assume without loss of generality that $k^*\in (i\dd k)$.
In particular, this means that $k-i\ge 2$ and $h(i+1,k-1)=h(k^*) = h(i,k)$,
Consequently, by the inductive assumption, $\GD(i,j)= \GD(i,k)+\GD(k,j)=\GD(i,k^*)+\GD(k^*,k)+\GD(k,j)\ge \GD(i,k^*)+\GD(k^*,j) \ge \GD(i,j)$, i.e., $\GD(i,j)=\GD(i,k^*)+\GD(k^*,j)$ holds as claimed.
Here, the first inequality holds because $k\in (k^*\dd j)\cap (K\cup \{k^*+1,k^*+2,j-2,j-1\})$,
whereas the second one is due to $k^*\in K$ (because $k^*$ is a valley).
\end{proof}

Our approximation algorithm (implemented as \cref{alg:3Dyck-approx}) computes $\BD$ that approximates $\GD$
in the same way  $\AD$ approximates $\D$ in \cref{alg:Dyck-Approx}. The only difference is that we use \cref{obs:hd,fct:gdp} (and the definition of $\GD$) to prune some states and transitions.
For each of the remaining states,
%\tknote{We do not have time to process all substrings.} the algorithm computes an value $\BD(i,j)\approx \GD(i,j)$.
If $h(i,j)<h(i+1,j-1)$, then \cref{alg:3Dyck-approx} relies on \cref{fct:gdp} and considers the smallest index $k\in (i\dd j)$ with $h(k)=h(i,j)$ as the sole potential pivot,
i.e, it returns $\BD(i,k)+\BD(k,j)$. If $h(i,j)=h(i+1,j-1)$, then \cref{alg:3Dyck-approx} mimics \cref{alg:Dyck-Approx}.

\begin{algorithm}[ht]
\SetKwBlock{Begin}{}{end}
	$\BD(i,j)$\Begin{
		\lIf{$j=i$}{\Return{$0$}}
		\lIf{$j=i+1$}{\Return{$1$}}
		\lIf{$h(i,j)<\max(h(i),h(j))-2d$}{\Return{$\infty$}}\label{ln:easydef}
		\If{$h(i,j) = h(i+1,j-1)$}{
			Select the smallest $k\in (i\dd j)$ such that $h(k)=h(i,j)$\;
			\Return{$\BD(i,k)+\BD(k,j)$}\;\label{ln:easy}
		}
		$c := \BD(i+1,j-1)+\dyck(x[i+1]x[j])$\;\label{ln:hard}
		$\tau_{i,j} := \tau\cdot 2^{\min(\nu(i),\nu(j)}$\;
		$K_{i,j} :=$ the set of $\tau_{i,j}$ smallest and $\tau_{i,j}$ largest elements of $K\cap (i\dd j)$\;
		\ForEach{$k\in K_{i,j}\cup(\{i+1,i+2,j-2,j-1\}\sm \{i,j\})$}{
			$c := \min(c, \BD(i,k)+\BD(k,j))$\;\label{ln:call}
		}
		\Return{$c$}\;
	}
	\caption{Recursive implementation of $\protect\BD(i,j)$}\label{alg:3Dyck-approx}
	\end{algorithm}
	
	The analysis of the approximation ratio of \cref{alg:3Dyck-approx} mimics that of \cref{alg:Dyck-Approx}.

	\begin{lemma}\label{lem:bd}
		If $\tau \ge 2$, then, for each $i,j\in [0\dd n]$ with $i\le j$, we have $\GD(i,j)\le \BD(i,j)$ and, if $\GD(i,j)\le d$, we further have $\BD(i,j)\le \GD(i,j) + \frac{8}{\tau} c_{i,j}\log c_{i,j}$, where $c_{i,j} :=|K\cap (i\dd j)|$, and we assume $0\log 0 = 0$.
	\end{lemma}
	\begin{proof}
		As for the upper bound, we proceed by induction on $j-i$. For $j-i\le 1$, we have $\BD(i,j)=\GD(i,j)$.
		For $j-i\ge 2$, the lower bound $\GD(i,j)\le \BD(i,j)$ follows directly form the definitions of $\BD$ and $\GD$.
		If $h(i,j)<\max(h(i),h(j))-2d$, the upper bound follows from \cref{obs:hd,lem:good}.
		If $h(i,j)=h(i+1,j-1)$, the upper bound follows form \cref{fct:gdp}.
		Otherwise, the upper bound follows directly from the definitions of $\BD$ and $\GD$ unless $\GD(i,j)=\GD(i,k)+\GD(k,j)$ for some $k\in (i\dd j)\cap K$.
		Let $r=\min(c_{i,k},c_{k,j})$ and let $i',j'$ be the smallest and the largest multiple of $2^{\lceil\log ((r+1)/\tau)\rceil}$ within $[i\dd j]$. (We have $k'\in (i'\dd j')$ due to $r+1 \ge \frac{2(r+1)}{\tau} \ge 2^{\lceil\log ((r+1)/\tau)\rceil}$.)
		Note that $\tau(i'-i) < \tau 2^{\lceil\log ((r+1)/\tau)\rceil} < 2(r+1)$, so $\tau(i'-i)\le 2r$;
		symmetrically,  $\tau(j-j')\le 2r$.
		Moreover, $k\in K_{i',j'}$. Thus, 
		%$\BD(i,j)\le (i'-i)+ \BD(i',j')+(j-j') \le  \BD(i',k)+\BD(k,j') + (i'-i+j-j') \le \D(i',k) + \frac{8}{\tau} c_{i',k}\log c_{i',k} + \D(k,j') + \frac{8}{\tau} c_{k,j'}\log c_{k,j'} +(i'-i+j-j') \le \D(i,k) + \frac{8}{\tau} c_{i,k}\log c_{i,k} + \D(k,j) + \frac{8}{\tau} c_{k,j}\log c_{k,j} +2(i'-i+j-j')	\le  \D(i,j) +\frac{8}{\tau} (c_{i,k}\log c_{i,k} +  c_{k,j}\log c_{k,j}+r)
		%\le  \D(i,j) + \frac{8}{\tau} c_{i,j}\log c_{i,j}$,
		%with the last inequality due to ${r=\min(c_{i,k},c_{k,j})\le \frac12(c_{i,k}+c_{k,j})\le \frac12c_{i,j}}$.
		
		\begin{align*}
		\BD(i,j)&\le (i'-i)+ \BD(i',j')+(j-j')\\
		 & \le  \tfrac{2r}{\tau}+\BD(i',k)+\BD(k,j') + \tfrac{2r}{\tau}\\
		 &\le \GD(i',k) + \tfrac{8}{\tau} c_{i',k}\log c_{i',k} + \GD(k,j') + \tfrac{8}{\tau} c_{k,j'}\log c_{k,j'} +\tfrac{4r}{\tau} \\
		 &\le (i'-i)+\GD(i,k) + \tfrac{8}{\tau} c_{i,k}\log c_{i,k} + \GD(k,j)+ (j-j')+ \tfrac{8}{\tau} c_{k,j}\log c_{k,j} +\tfrac{4r}{\tau}	\\
		 &\le  \GD(i,j) +\tfrac{8}{\tau} (c_{i,k}\log c_{i,k} +  c_{k,j}\log c_{k,j}+r)\\
		 &=  \GD(i,j) +\tfrac{8}{\tau} (\max(c_{i,k},c_{k,j})\log \max(c_{i,k},c_{k,j}) +  \min(c_{i,k},c_{k,j})\log (2\min(c_{i,k},c_{k,j})))\\
		 &\le \GD(i,j) + \tfrac{8}{\tau} (\max(c_{i,k},c_{k,j})\log c_{i,j} + \min(c_{i,k},c_{k,j}) \log c_{i,j}) \\
		 &\le  \GD(i,j) + \tfrac{8}{\tau} c_{i,j}\log c_{i,j}.\qedhere\end{align*}
		
	\end{proof}

On the other hand, the complexity analysis is more complex than in \cref{sec:PTAS}: it involves a charging argument bounding the number of states processed using the insight of \cref{fct:gdp}.
\begin{proposition}\label{prp:3approxdyck}
	There is an algorithm that, given a string $x\in \Sigma^n$, a threshold $d\in [1\dd n]$, and a parameter $\epsilon\in (0,1)$, in $\tOh(\epsilon^{-1}n d)$ time reports $\GD(0,n)>d$ or outputs a value $v$ such that $\GD(0,n)\le v \le (1+\epsilon)\GD(0,n)$.
\end{proposition}
\begin{proof}
	In the preprocessing, we use \cref{fct:preprocess} in order to guarantee that there are at most $2\dyck(x)$ valleys and thus $|K|\le 6\dyck(x)$.  Then, we construct a data structure that, given $i,j\in[0\dd n]$, reports the smallest $k\in [i\dd j]$ such that $h(k)=h(i,j)$~\cite{RMQ}.
	Finally, we run $\BD(0,n)$ with $\tau = \lceil 48 \epsilon^{-1} \log |K|\rceil$
	and memoization of the results of recursive calls.
	By \cref{lem:bd}, the returned value satisfies $\GD(0,n)\le \BD(0,n) \le \GD(0,n)+\frac{8}{\tau} |K|\log |K| \le (1+\epsilon)\GD(0,n) $.
	
	The running time analysis is more complex than in the proof of \cref{thm:approxdyck}.
	We say that a call $\BD(i,j)$ is \emph{hard} if it reaches \cref{ln:hard},
	\emph{easy} if it terminates at \cref{ln:easydef} or \cref{ln:easy}, and \emph{trivial} otherwise. Observe that the total cost of trivial calls is $\tOh(n)$.
	Moreover, the cost of each easy call is $\tOh(1)$ plus the cost of the two calls made in \cref{ln:easy},
	but the call $\BD(i,k)$ is never easy (it can be hard or trivial).
	This is because the choice of $k$ as the smallest index in $(i\dd j)$ with $h(k)=h(i,j)$ guarantees that either $k=i+1$ or $h(i+1,k-1)>h(k)=h(i,k)=h(i,j)\ge \max(h(i),h(j))-2d \ge h(i)-2d = \max(h(i),h(k))-2d$.
	Consequently, the cost of each easy call can be charged to its parent or sibling (which is hard or trivial),
	and it suffices to bound the total running time of hard calls.
	By symmetry, we only bound the cost of calls $\BD(i,j)$ with $h(i)\ge h(j)$.
	We then observe that if $h(i)\ge h(j)=h(j')$ and $i> j > j'$, then $\BD(i,j')$ is easy.
	Consequently, there are at most $4d+1$ hard calls per $i$.
	The cost of each hard call $\BD(i,j)$ is $\tOh(\tau_{i,j})=\tOh(\tau 2^{\nu(i)})$,
	for a total of $\tOh(d\tau \sum_{i=0}^n 2^{\nu(i)})=\tOh(\epsilon^{-1} nd)$.
	\end{proof}

	\begin{theorem}\label{thm:3approxdyck}
		There is an algorithm that, given a string $x\in \Sigma^n$, a threshold $d\in [1\dd n]$, and a parameter $\epsilon\in (0,1)$, in $\tOh(\epsilon^{-1}nd)$ time reports $\dyck(x)>d$ or outputs a value $v$ such that $\dyck(x)\le v \le (3+\epsilon)\dyck(x)$.
	\end{theorem}
	\begin{proof}
		We apply \cref{prp:3approxdyck} with adjusted $d$ (three times larger) and $\epsilon$ (three times smaller).
		The correctness follows from \cref{lem:good}.
	\end{proof}

\section{Constant Approximation in Sub-quadratic Time}\label{sec:consub}

In this section we provide an algorithm that given a string $x$ over some alphabet $\Sigma$, computes constant approximation of $\dyck(x)$ in sub quadratic time. Formally we show the following.

\begin{theorem}\label{thm:mainapproxsubquadratic}
Given a string $x$ of length $n$ over an alphabet $\Sigma$, there exists a randomized algorithm Dyck-Est that provides an upper bound on $\dyck(x)$ that is at most a fixed constant multiple of $\dyck(x)$, in time $\tOh(n^{1.971})$ with probability at least $1-1/n^{-9}$.
\end{theorem}

Instead of directly proving the above theorem we prove the following gap version.

\begin{theorem}\label{thm:gapapproxsubquadratic}
Given a string $x$ of length $n$ over an alphabet $\Sigma$, 
for every $\theta \in[n^{-1/34},1]$, there exists a randomized algorithm Gap-Dyck-Est$_\theta$ that on input $x$ provides an estimation $t$ in time $\tOh(n^{1.971})$, where $t\ge \dyck(x)$ and if $\dyck(x)\le \theta n$ then $t\le k \theta n$ (here $k$ is a fixed constant) with probability at least  $1-1/n^{-9}$.
\end{theorem}

\begin{proof}[Proof of Theorem~\ref{thm:mainapproxsubquadratic} from Theorem~\ref{thm:gapapproxsubquadratic}]
Run the algorithm from Theorem~\ref{thm:3approxdyck} with parameter $\epsilon=1$ on $x$ for $\tOh(n^{1.971})$ time. If it terminates then it outputs $4$ approximation of Dyck edit distance of $x$. Otherwise, if it fails then we know $\dyck(x)\ge n^{33/34}$. Now run Gap-Dyck-Est$_\theta$ for $\theta \in \{n^{-1/34},2n^{-1/34},\dots,1\}$ and output the minimum of all estimations obtained. Let $j$ be the smallest index such that $\theta_j n\ge \dyck(x)$. Then the output is at most $k\theta_j n\le 2k \dyck(x)$. As there are $\log n$ different values of $\theta$, the total running time of Dyck-Est is $\tOh(n^{1.971})$.
\end{proof}

The rest of the section is devoted to prove Theorem~\ref{thm:gapapproxsubquadratic}.

\subsection{Window Decomposition}\label{sec:window}

We define a window $w\subseteq [n]$ of size $d$ to be an interval in  $[n]$ having length $d$. Let $s(w)$ denote the starting index of $w$ in $s$ and $e(w)$ denote the last index of $w$ in $s$. Given a window $w=\{i_1,\dots,i_\ell\}$, $x[w]$ denotes the substring of $x$ that starts at index $i_1$ and ends at index $i_\ell$. 

A \emph{window pair} is a pair of windows $(w,w')$, 
and a weighed window pair is a triple $(w,w',c)$ such that $(w,w')$ is a window pair and $c\in \mathbb{R}_{\ge 0}$ is a weight. Given a window pair $(w,w')$, the cost of $(w,w')$ is $\dyck(x[w]\circ x[w'])$. For a weighted window pair $(w,w',c)$ we call it certified if $c\ge \dyck(x[w]\circ x[w'])$.

\begin{definition}\label{def:consistent}
A set $\{(w_1,w'_1),\ldots, (w_\ell,w'_\ell)\}$ of window pairs is a \emph{consistent decomposition}
of $S=\bigcup_{i=1}^{\ell} (w_i\cup w'_i)$ if  the $2\ell$ windows are disjoint 
and $\{(s(w_i),s(w'_i)):i\in [1\dd \ell]\}$ forms a non-crossing matching.
We also lift this definition to sets of weighted window pairs.
\end{definition}

\paragraph{Large windows.}
Given a string $x$ of length $n$ and a large window size parameter $s_1\in [1\dd n]$, we start our algorithm by introducing variable sized windows $\mathcal{J}$:

\[\mathcal{J} = \{w \sub [1\dd n]\; :\; 0\le |w|\le 5s_1\text{ and } \theta s_1 \mid s(w)-1\text{ and }\theta s_1\mid e(w)\}.\]

The following claim directly follows from the construction.

\begin{claim}\label{claim:sizelarge}
$|\mathcal{J}|=O(\frac{n}{\theta^2 s_1})$
\end{claim}

%\tknote{The following lemma summarizes the next two paragraphs.}

\noindent Next we show set $\mathcal{J}$ satisfies the following property.
\lemlarge*

\paragraph{Large alignment window decomposition.}

 Given any alignment $M$ and a parameter $s_1>0$, we now describe a process generating a set $\mathcal{S}_M$ where each element in the set is a pair of windows that are matched by $M$ and together these windows cover the whole input string $x$. Moreover, each window has length at most $4s_1$.

We construct $\mathcal{S}_M$ by using a recursive process $P_M$.
At any recursion step given an interval $[i_1\dd i_2]\subseteq [n]$ we do one of the followings. 

\begin{enumerate}

\item If $i_2-i_1\le 4s_1$, define a window pair $t=([i_1\dd \lfloor \frac{i_1+i_2}{2}\rfloor],[\lfloor \frac{i_1+i_2}{2}\rfloor+1\dd i_2])$. Add this pair to $\mathcal{S}_M$. 

Otherwise, let $k_1=\max_{\begin{subarray}jj\in[i_1\dd i_1+s_1]\\ M(j)\in[i_2-s_1\dd i_2]\end{subarray}} j$. If there does not exist such an index set $k_1=0$.

\item Let first consider the case where $k_1=0$. If there exists a pivot $p\in [i_1+s_1\dd i_1+2s_1]$ (If many consider the leftmost one)
define a window pair $t=([i_1\dd \lfloor \frac{i_1+p}{2}\rfloor],[\lfloor \frac{i_1+p}{2}\rfloor+1\dd p])$. Add this pair to $\mathcal{S}_M$ and recurse on interval $[p+1\dd i_2]$. 

\item Otherwise, if there exists a pivot $p\in [i_2-2s_1\dd i_2-s_1]$ (If many consider the rightmost one)
define a window pair $t=([p+1\dd \lfloor \frac{p+i_2}{2}\rfloor],[\lfloor \frac{p+i_2}{2}\rfloor+1\dd i_2])$. Add this pair to $\mathcal{S}_M$ and recurse on interval $[i_1\dd p]$. 

\item Otherwise, consider the left most pivot $p\ge i_1+2s_1$. Recurse separately on  $[i_1\dd p]$ and $[p+1\dd i_2]$. Notice there must exist one such pivot as we assume $\nexists p\in[i_1+s_1\dd i_1+2s_1]$ and $k_1=0$. 

Thus, from now on wards we assume $k_1\neq 0$.

\item If there exists a pivot $p\in [i_1+s_1\dd i_1+2s_1]$ for the subproblem defined on $[k_1\dd M(k_1)]$ (if many, consider the left most one), define a window pair $t=([i_1\dd p],[M(k_1)\dd i_2])$. Add this pair to $\mathcal{S}_M$ and recurse on interval $[p+1\dd M(k_1)-1]$. 

\item Otherwise, if there exists a pivot $p\in [i_2-2s_1\dd i_2-s_1]$ for the subproblem defined on $[k_1\dd M(k_1)]$ (if many, consider the right most one), define a window pair $t=([i_1\dd k_1],[p+1\dd i_2])$. Add this pair to $\mathcal{S}_M$ and recurse on interval $[k_1+1\dd p]$.

\item  Otherwise, if either $k_1\ge i_1+w$ or $M(k_1)\le i_2-s_1$ (or both), define a window pair $t=([i_1\dd k_1],[M(k_1)\dd i_2])$. Add this pair to $\mathcal{S}_M$ and recurse on $[k_1+1\dd M(k_1)-1]$.

\item Otherwise, consider the leftmost pivot $p\ge i_1+2s_1$. 
 Define a window pair $t=([i_1\dd k_1],\allowbreak [M(k_1)\dd i_2])$. Add this pair to $\mathcal{S}_M$ and recurse separately on  $[k_1\dd p]$ and $[p+1\dd M(k_1)]$.
Notice such a pivot always exists as $\exists j\in[k_1+1\dd i_1+s_1]$ such that $M(j)<i_2-2s_1$ and $i_2-i_1>4s_1$.

\end{enumerate}

%Let $[i_1,i_2]$ be some interval that are considered by process $P_M$ at some recursion step. We call this interval to be \emph{far-pivot} if there does not exists any pivot $p$ valid for $[i_1,i_2]$ with respect to $M$ in the interval $[i_1,i_2]$

\begin{lemma}\label{lem:decomposition}
Let $x$ be a string of length $n$. Given an alignment $M$ of $x$ and a parameter $s_1\in [1\dd n]$ let $\mathcal{S}_M$ be the set generated by process $P_M$. We can make the following claims:
\begin{enumerate}
    \item each window pair $(w_i,w'_i)\in \mathcal{S}_M$ satisfies $|w_i|,|w'_i|\le 4s_1$,
    \item $\mathcal{S}_M$ is a consistent decomposition of $[1\dd n]$,
    \item $\cost_M(x)\ge \sum_{(w_i,w'_i)\in \mathcal{S}_M} \cost_M(x[w_i]\circ x[w'_i])$,
    \item $|\mathcal{S}_M|=O(\frac{n}{s_1})$.
\end{enumerate}
\end{lemma}

 \begin{proof}

 By construction, we can ensure that every window generated has length at most $4s_1$.

\vspace {2mm}

 Again by construction we can claim that every index in $[n]$ is contained in some window pair.
Moreover, once a window pair $(w_i,w'_i)$
is created they are removed from all the intervals that are considered in the future recursions. Thus, any two window pairs in $\mathcal{S}_M$ are disjoint and $\mathcal{S}_M$ is a decomposition of $[1\dd n]$. 
Now consider any two distinct pair of windows $(w_i,w'_i)$ and $(w_j,w'_j)$. Consider the recursion tree $T$ that is induced by process $P_M$. Each node of this tree is represented by an interval $[i_1\dd i_2]$ processed by $P_M$. Moreover, if $[i_1\dd i_2]$ was created while processing interval $[\tilde{\imath}^1\dd \tilde{\imath}^2]$, then we make it's associated node the parent of $[i_1\dd i_2]$. Notice every window pair can be associated with a node in $T$. Let $n_i,n_j$ be the associated node of $(w_i,w'_i)$ and $(w_j,w'_j)$. If $n_i$ is an ancestor of $n_j$ then the consistency can be argued using condition 3. Otherwise, if $n_j$ is an ancestor of $n_i$ then the consistency can be argued using condition 4. Otherwise, let $n_a$ be the lowest common ancestor of $n_i,n_j$ and $[i_1\dd i_2]$ be the associated interval. Then if $n_i$ appears left to $n_j$ then both $w_i, w'_i$ are in the interval $[i_1\dd p]$ for case 4 (or in $[k_1\dd p]$ for case 8) and $w_j, w'_j$ are in the interval $[p+1\dd i_2]$ for case 4 (or in $[p+1\dd M(k_1)]$ for case 8). Thus, the consistency can be argued using condition 1. Similarly, if $n_j$ appears left to $n_i$ consistency can be argued using condition 2.

\vspace {2mm}

To prove the third point we claim, for any index $j\in[n]$ if $M(j)\neq \bot$ and $j\in w_i$, then $M(j)\in w_i\cup w'_i$, where $(w_i,w'_i)\in \mathcal{S}_M$. Thus, we show $\cost_M(x)\ge \sum_{(w_i,w'_i)\in \mathcal{S}_M} \cost_M(x[w_i]\circ x[w'_i])$. Now we prove the claim. 
In the recursion process let the window pair $(w_i,w'_i)$ is created while considering the interval $[i_1\dd i_2]$. For contradiction let $M(j)\neq \bot$, $j\in w_i$ but $M(j)\notin w_i\cup w'_i$. 
This trivially can't happen if $i_2-i_1\le 4s_1$ as in that case $M(j)\in [i_1\dd i_2]=w_i\cup w'_i$.
Next consider the case where $k_1=0$. Here as $p$ is a valid pivot
both $j,M(j)\in [i_1\dd p]=w_i\cup w'_i$ (or $j,M(j)\in [p+1\dd i_2]=w_i\cup w'_i$). 

Next consider the case where $w_i=[i_1\dd p]$ and $w'_i=[M(k_1)\dd i_2]$. Here if $j\le k_1$ then $M(j)\in [i_1\dd k_1]\cup [M(k_1)\dd i_2]$. Hence, we get a contradiction. Otherwise, $j$ is contained in the subproblem defined on $[k_1\dd M(k_1)]$. But as for this, $p$ is a valid pivot $M(j)\in[k_1\dd p]\subseteq w_i$. Thus, we get a contradiction.
The case where $w_i=[i_1\dd k_1]$ and $w'_i=[p+1\dd i_2]$ or $w_i=[i_1\dd k_1]$ and $w'_i=[M(k_1)\dd i_2]$ can be argued in a similar way.

\vspace {2mm}

To put an upper bound on the size of $\mathcal{S}_M$, consider the recursion tree $T$ that is induced by process $P_M$. Each node of this tree is represented by an interval $[i_1\dd i_2]$ processed by $P_M$. Moreover, if $[i_1\dd i_2]$ was created while processing interval $[\tilde{\imath}^1\dd \tilde{\imath}^2]$, then we make it's associated node the parent of $[i_1\dd i_2]$. Notice $|\mathcal{S}_M|\le |T|$. Thus, we now try to bound the size of this recursion tree. 

For this we first claim that each leaf node in $T$ is associated with a pair of windows $(w_i,w'_i)$ such that $|w_i|+|w'_i|\ge s_1$. Let $[i_1\dd i_2]$ be the associated interval. Trivially $i_2-i_1\le 4s_1$. We further show $i_2-i_1\ge s_1$, thus proving our claim. 
As otherwise assume $i_2-i_1< s_1$. Let $[\tilde{\imath}^1\dd \tilde{\imath}^2]$ be the interval associated with the parent node. By construction $\tilde{\imath}^2-\tilde{\imath}^1>4s_1$. While processing this interval let us first consider the case where $k_1=0$. If there exists a pivot $p\in [\tilde{\imath}^1+s_1\dd \tilde{\imath}^1+2s_1]$ (or $\in [\tilde{\imath}^2-2s_1\dd \tilde{\imath}^2-s_1]$) then by construction $i_2-i_1\ge \tilde{\imath}^2-\tilde{\imath}^1-2s_1>2s_1$, and thus we get a contradiction. In the other case as the pivot $p$ is at a distance $>2s_1$ from both $\tilde{\imath}^1,\tilde{\imath}^2$, we have $i_2-i_1>2s_1$.
Next consider the case where $k_1\neq 0$. In this case if there exists a pivot $p\in [\tilde{\imath}^1+s_1\dd \tilde{\imath}^1+2s_1]$ (or $\in [\tilde{\imath}^2-2s_1\dd \tilde{\imath}^2-s_1]$) then a window pair covering an interval of length at most $3s_1$ is generated and we recurse on the rest of the interval. Thus, $i_2-i_1>s_1$. Otherwise, no such pivot exists total length of the window pair is $\le 2s_1$. Thus, again $i_2-i_1>2s_1$. In the last case as $k_1<i_1+s_1$ and $p\ge i_1+2s_1$, $p-k_1\ge w$. Similarly, we can show $M(k_1)-(p+1)\ge w$ and hence $i_2-i_1\ge s_1$.
Following the claim, the maximum number of leaf nodes in $T$ is at most $\frac{n}{s_1}$. 

Next we claim that the interval corresponding to every internal node satisfies one of the following: (i) either the node has degree 2 (ii) or while processing the corresponding interval, $P_M$ generates a pair of windows of total length at least $s_1$. Notice as we already argued that any two window pairs generated by $P_M$ are completely disjoint, combing this with the fact that the total number of leaf nodes in $T$ is at most $\frac{n}{s_1}$, we can prove $|T|=O(\frac{n}{s_1})$. Notice the claim is trivial from our construction as in every case except $4,8$, $P_M$ generates a pair of windows of total length at least $s_1$. On the other hand in case $4,8$ the corresponding nodes have degree 2.
\end{proof}

\paragraph{Capped large alignment window decomposition.}

After creating set of window pairs $\mathcal{S}_M$, we cap all the windows using set $\mathcal{J}$ to constructed set $\mathcal{\tilde{S}}_M$. Formally for each window pair $(w_i,w'_i)\in \mathcal{S}_M$ we add one window pair $(\tilde{w}_i,\tilde{w}'_i)$ to $\mathcal{\tilde{S}}_M$ where the starting index and the size of each window is a multiple of $\theta s_1$. Formally $s(\tilde{w}_i)=\lceil \frac{s(w_i)}{\theta s_1} \rceil \cdot \theta s_1+1$ and $e(\tilde{w}_i)=\lceil \frac{e(w_i)}{\theta s_1} \rceil \cdot \theta s_1$. 
If $s(\tilde{w}_i)>e(\tilde{w}_i)$, set $\tilde{w}_i=\emptyset$. Notice in this case $|w_i|=0$.
Similarly, we can define $\tilde{w}'_i$. 
Notice $\mathcal{\tilde{S}}_M\subseteq \mathcal{J}\times \mathcal{J}$.

\begin{lemma}\label{lem:cappedlargedecomposition}
Let $x$ be a string of length $n$. Given an alignment $M$ of $x$ and a parameter $s_1\in [1\dd n]$ let $\mathcal{\tilde{S}}_M$ be the set generated by capped large alignment window decomposition. We can make the following claims:

\begin{enumerate}
\item each window pair $(\tilde{w}_i,\tilde{w}'_i)\in \mathcal{\tilde{S}}_M$ satisfies $|\tilde{w}_i|,|\tilde{w}'_i|\le 5s_1$,
\item $\mathcal{\tilde{S}}_M$ is a consistent decomposition of $[1\dd n]$,
   \item $\sum_{(\tilde{w}_i,\tilde{w}'_i)\in \mathcal{\tilde{S}}_M} \dyck(x[\tilde{w}_i]\circ x[\tilde{w}'_i])\le \cost_M(x)+O(\theta n)$,
\item $|\mathcal{\tilde{S}}_M|=O(\frac{n}{s_1})$.
\end{enumerate}

\end{lemma}

\begin{proof}
For each window pair $(\tilde{w}_i,\tilde{w}'_i)\in \mathcal{\tilde{S}}_M$ there exists a corresponding window pair $({w}_i,{w}'_i)\in \mathcal{S}_M$ such that $|{w}_i|, |{w}'_i|\le 4s_1$. Now we define $\tilde{w}_i$ to be the window generated from ${w}_i$ where  $s(\tilde{w}_i)=\lceil \frac{s(w_i)}{\theta s_1} \rceil \cdot \theta s_1+1$ and $e(\tilde{w}_i)=\lceil \frac{e(w_i)}{\theta s_1} \rceil \cdot \theta s_1$. Thus, $|\tilde{w}_i|\le |w_i|+\theta s_1\le 5s_1$. Similarly, we can claim for $|\tilde{w}'_i|$.

\vspace {2mm}

First we claim for any pair of windows $(\tilde{w}_i,\tilde{w}'_i)$ and $(\tilde{w}_j,\tilde{w}'_j)$, $(\tilde{w}_i\cup \tilde{w}'_i)\cap (\tilde{w}_j\cup \tilde{w}'_j)=\emptyset$. We first show $\tilde{w}_i\cap \tilde{w}_j=\emptyset$. The other three options can be argued in a similar way. 
Let $w_i, w_j$ be the corresponding window in $\mathcal{S}_M$. Without loss of generality, assume $e(w_i)<s(w_j)$. We claim $e(\tilde{w}_i)<s(\tilde{w}_j)$. Otherwise, assume $e(\tilde{w}_i)>s(\tilde{w}_j)$. Note they can not be equal as $\theta s_1|e(\tilde{w}_i)$ and $\theta s_1|s(\tilde{w}_j)-1$. Thus, $e(\tilde{w}_i)-s(\tilde{w}_j)\ge \theta s_1-1$. Now by rounding $e(\tilde{w}_i)-e({w}_i)\le \theta s_1-1$. Thus, $e(w_i)\ge s(\tilde{w}_j)\ge s(w_j)$ and we get a contradiction.

For arguing consistency consider any two pair of windows $(w_i,w'_i)$ and $(w_j,w'_j)$. Now notice if any of the four windows are an empty set then the consistency follows trivially. Otherwise, it follows from the following fact that we claimed above. For any two windows $\tilde{w}_i,\tilde{w}_j$ in $\mathcal{\tilde{S}}_M$ if $w_i,w_j$ be the corresponding windows in $\mathcal{S}_M$ then $e(w_i)<s(w_j)$, implies $e(\tilde{w}_i)<s(\tilde{w}_j)$ and vice versa.

\vspace {2mm}

Let $(\tilde{w}_i,\tilde{w}'_i)$ be the window pairs generated from $(w_i,w'_i)\in \mathcal{S}_M$. Notice $(w_i\cup w'_i)\setminus (\tilde{w}_i\cup \tilde{w}'_i)\le 2 \theta s_1$. Thus, total number of indices $i \in w_i\cup w'_i$ such that either $i$ or $M(i)\in (w_i\cup w'_i)\setminus (\tilde{w}_i\cup\tilde{w}'_i)$ are at most $4 \theta s_1$. Thus, $\dyck(x[\tilde{w}_i]\circ x[\tilde{w}'_i])\le \cost_M(x[w_i]\circ x[w'_i]) +4\theta s_1$.
Moreover, as by \cref{lem:decomposition} $\cost_M(x)\ge \sum_{(w_i,w'_i)\in \mathcal{S}_M} \cost_M(x[w_i]\circ x[w'_i])$ and $|\mathcal{S}_M|=O(\frac{n}{s_1})$; thus
we get $\sum_{(\tilde{w}_i,\tilde{w}'_i)\in \mathcal{\tilde{S}}_M} \dyck(x[\tilde{w}_i]\circ x[\tilde{w}'_i])\le \cost_M(x)+O(\theta n)$.

\vspace {2mm}

The claim directly follows as $|\mathcal{\tilde{S}}_M|\le |\mathcal{S}_M|$. 
\end{proof}

\iffalse
Notice $\mathcal{\tilde{S}}_M\subseteq \mathcal{J}\times \mathcal{J}$.
Moreover, $\forall (\tilde{w}_i,\tilde{w}'_i)\in \mathcal{\tilde{S}}_M$, $|\tilde{w}_i|,|\tilde{w}'_i|\le 5s_1$. 

Let $\mathcal{I}_M=\bigcup_{(\tilde{w}_i,\tilde{w}'_i)\in \mathcal{\tilde{S}}_M} \tilde{w}_i\cup \tilde{w}'_i$. We define an alignment $\tilde{M}$ of $x$ by restricting alignment $M$ to the indices in $\mathcal{I}_M$. We can make the following claim.

\begin{claim}
$\cost_{\tilde{M}}(x)\le \cost_M(x)+O(\theta n)$.
\end{claim}

\begin{proof}
 Let $(\tilde{w}_i,\tilde{w}'_i)$ be the window pair generated from $(w_i,w'_i)\in \mathcal{S}_M$. Notice $\tilde{w}_i\cup \tilde{w}'_i \subseteq w_i\cup w'_i$. Thus, for any pair of windows $(\tilde{w}_i,\tilde{w}'_i)$ and $(\tilde{w}_j,\tilde{w}'_j)$, $(\tilde{w}_i\cup \tilde{w}'_i)\cap (\tilde{w}_j\cup \tilde{w}'_j)=\emptyset$. Also, $(w_i\cup w'_i)\setminus (\tilde{w}_i\cup \tilde{w}'_i)\le 4 \theta s_1$. Thus, total number of indices $i \in w_i\cup w'_i$ such that either $i$ or $M(i)\in (w_i\cup w'_i)\setminus (\tilde{w}_i\cup\tilde{w}'_i)$ are at most $8 \theta s_1$. Moreover, as by \cref{lem:decomposition} $\cost_M(x)\ge \sum_{(w_i,w'_i)\in \mathcal{S}_M} \cost_M(x[w_i]\circ x[w'_i])$ and $|\mathcal{S}_M|=O(\frac{n}{s_1})$
we get $cots_{\tilde{M}}(x)\le \cost_M(x)+O(\theta n)$.
\end{proof}
\fi

Thus, given a string $x$ of length $n$, for any optimal alignment $M$, there exists a set $\mathcal{\tilde{S}}_M \subseteq \mathcal{J}\times \mathcal{J}$ such that $\mathcal{\tilde{S}}_M$ is a consistent decomposition of $[1\dd n]$ and 
$\sum_{(\tilde{w}_i,\tilde{w}'_i)\in \mathcal{\tilde{S}}_M} \dyck(x[\tilde{w}_i]\circ x[\tilde{w}'_i])\le \cost_M(x)+O(\theta n)$. 
As $M$ is optimal, this proves Lemma~\ref{lem:large}.

Next our objective is to find an estimation of $\dyck(x[\tilde{w}_i]\circ x[\tilde{w}'_i])$.
For this we further divide the large windows into smaller windows as follows. 

%For this we consider a subset of window pairs from $\mathcal{W}_L\subseteq \mathcal{J}\times \mathcal{J}$, such that for each optimal alignment $M$, $\mathcal{\tilde{S}}_M \subseteq \mathcal{W}_L$ and for each window pair $(w_i,w'_i)\in \mathcal{W}_L$ we compute an estimation of $cost(x[w_i]\circ x[w'_i])$. Then we combine all these sub-solutions using a dynamic program algorithm to obtain the cost estimation of $cost(x)$. 

\paragraph{Small windows.} 
Given a string $x$ of length $n$ and a small window size parameter $s_2\in [1\dd n]$, where $s_2|s_1$,
we further introduce variable size windows $\mathcal{K}$:

\[\mathcal{K}=\{w\sub [1\dd n] : |w|\le 5s_2\text{ and }\theta s_2 \mid s(w)-1\text{ and }\theta s_2 \mid e(w)\}.\]

For each window $w\in \mathcal{J}$ let $\mathcal{K}_w \subseteq \mathcal{K}$ be the set consisting of all the windows in $\mathcal{K}$ that are part of interval $w$. The following claim directly follows from the construction.

\begin{claim}\label{claim:sizesmall}
$|\mathcal{K}|=O(\frac{n}{\theta^2 s_2})$
\end{claim}

\noindent Next we show set $\mathcal{K}$ satisfies the following property.

\begin{lemma}\label{lem:small}
	Consider a string $x$ of length $n$, and parameters $s_1,s_2,\theta$ such that $\theta s_2 \mid \theta s_1$
	and $\theta s_1 \mid n$.
	For every window pair $(w,w')\in \mathcal{J}\times \mathcal{J}$,
	the set $w\cup w'$ can be decomposed into a consistent set of window pairs $\mathcal{S}\sub \mathcal{K}\times \mathcal{K}$ such that $\sum_{(q,q')\in \mathcal{S}} \dyck(x[q]\circ x[q']) \le \dyck(x[w]\circ x[w'])+\Oh(\theta |w\cup w'|)$.
\end{lemma}

\paragraph{Small alignment window decomposition.} 

Given a string $x$, 
with small window size parameter $s_2$, for each pair $(w,w')\in \mathcal{J}\times \mathcal{J}$, and each alignment $\tilde{M}$ of $x[w]\circ x[w']$, we use process $P_{\tilde{M}}^{(w,w')}$ that is similar to the process used in \emph{large alignment window decomposition} to further divide interval $[w\cup w']$ with respect to alignment $\tilde{M}$ and parameter $s_2$ for generating a set $\mathcal{S}_{\tilde{M}}^{(w,w')}$ containing pairs of windows. Similar to \cref{lem:decomposition} we can argue the following.

\begin{lemma}\label{lem:decompositionsmall}
Let $x$ be a string of length $n$.
For every window pair $(w,w')\in \mathcal{J}\times \mathcal{J}$, given an alignment $\tilde{M}$ of $x[w]\circ x[w']$
and parameters $s_1,s_2\in [1\dd n]$, let $\mathcal{S}_{\tilde{M}}^{(w,w')}$ be the set generated by process $P_{\tilde{M}}^{(w,w')}$. Thus, we can make the following claims.
\begin{enumerate}
    \item Each window pair $(q_i,q'_i)\in \mathcal{S}_{\tilde{M}}^{(w,w')}$ satisfies $|q_i|,|q'_i|\le 4s_2$,
    \item $\mathcal{S}_{\tilde{M}}^{(w,w')}$ is a consistent decomposition of $w\cup w'$,
    \item $\cost_{\tilde{M}}(x[w]\circ x[w'])\ge \sum_{(q_i,q'_i)\in \mathcal{S}_{\tilde{M}}^{(w,w')}} \cost_{\tilde{M}}(x[q_i]\circ x[q'_i])$,
    \item $|\mathcal{S}_{\tilde{M}}^{(w,w')}|=O(\frac{(|w|+|w'|)}{s_2})$.
\end{enumerate}
\end{lemma}

\paragraph{Capped small alignment window decomposition.} 
For each $(w,w')\in \mathcal{J}\times \mathcal{J}$, after creating set $\mathcal{S}_{\tilde{M}}^{(w,w')}$, we cap all the windows using set $\mathcal{K}$ to construct set $\mathcal{\tilde{S}}_{\tilde{M}}^{(w,w')}$. Formally for each window pair $(p_i,p'_i)\in \mathcal{S}_{\tilde{M}}^{(w,w')}$ we add one window $(\tilde{p}_i,{\tilde{p}}'_i)$ to $\mathcal{\tilde{S}}_{\tilde{M}}^{(w,w')}$ where the starting index and the size of each window is a multiple of $\theta s_2$. Formally $s(\tilde{p}_i)=\lceil \frac{s(w_i)}{\theta s_2} \rceil\cdot \theta s_2+1$ and $e(\tilde{p}_i)=\lceil \frac{e(w_i)}{\theta s_2} \rceil\cdot \theta s_2$. Similarly, we can define for $\tilde{p}'_i$. Notice by construction $\mathcal{\tilde{S}}_{\tilde{M}}^{(w,w')}\subseteq \mathcal{K}\times \mathcal{K}$. Similar to \cref{lem:cappedlargedecomposition} we can argue the following.

\begin{lemma}\label{lem:cappedsmalldecomposition}
Let $x$ be a string of length $n$. 
For every window pair $(w,w')\in \mathcal{J}\times \mathcal{J}$, given an alignment $\tilde{M}$ of $x[w]\circ x[w']$
and parameters $s_1,s_2\in [1\dd n]$,
let $\mathcal{\tilde{S}}_{\tilde{M}}^{(w,w')}$ be the set generated by capped small alignment window decomposition. We can make the following claims:
\begin{enumerate}
\item Each window pair $(\tilde{q}_i,\tilde{q}'_i)\in \mathcal{\tilde{S}}_{\tilde{M}}^{(w,w')}$ satisfies $|\tilde{q}_i|,|\tilde{q}'_i|\le 5s_2$,
\item $\mathcal{\tilde{S}}_{\tilde{M}}^{(w,w')}$ is a consistent decomposition of $w\cup w'$,
   \item $\sum_{(\tilde{q}_i,\tilde{q}'_i)\in \mathcal{\tilde{S}}_{\tilde{M}}^{(w,w')}} \dyck(x[\tilde{q}_i]\circ x[\tilde{q}'_i])\le \cost_{\tilde{M}}(x[w]\circ x[w'])+O(\theta |w\cup w'|)$,
\item $|\mathcal{\tilde{S}}_{\tilde{M}}^{(w,w')}|=O(\frac{(|w|+|w'|)}{s_2})$.
\end{enumerate}

\end{lemma}

Thus, given a string $x$ of length $n$, for every window pair $(w,w')\in \mathcal{J}\times \mathcal{J}$ and for every optimal alignment $\tilde{M}$ of $x[w]\circ x[w']$ there exists a set $\mathcal{\tilde{S}}_{\tilde{M}}^{(w,w')} \subseteq \mathcal{K}\times \mathcal{K}$ such that $\mathcal{\tilde{S}}_{\tilde{M}}^{(w,w')}$ is a consistent decomposition of $w\cup w'$ and 
$\sum_{(\tilde{q}_i,\tilde{q}'_i)\in \mathcal{\tilde{S}}_{\tilde{M}}^{(w,w')}} \dyck(x[\tilde{q}_i]\circ x[\tilde{q}'_i])\le \cost_{\tilde{M}}(x[w]\circ x[w'])+O(\theta |w\cup w'|)$. 
As $\tilde{M}$ is optimal, this proves Lemma~\ref{lem:small}.

In the next section we design an algorithm that for each window pair $(w,w')\in \mathcal{J}\times \mathcal{J}$, either directly computes an estimation of $\dyck(x[w]\circ x[w'])$ or computes estimation of $\dyck(x[q]\circ x[q'])$ for each $(q,q')\in \mathcal{\tilde{S}}_{\tilde{M}}^{(w,w')}$, and then we combine these estimations using a dynamic program algorithm to estimate $\dyck(x)$.

\subsection{Estimating Cost of Window Pairs}
\label{sec:cost_est}

In this section we describe Algorithm CertifyingWindowPairs(), that computes sets of $\mathcal{W}_L\subseteq \mathcal{J}\times \mathcal{J}$ and $\mathcal{W}_S\subseteq \mathcal{K}\times \mathcal{K}$
such that for every certified pair the algorithm provides a constant approximation of its cost. 
Moreover, we show $\mathcal{W}_L \cup \mathcal{W}_S$ contains a subset
that forms a consistent decomposition of interval $[1\dd n]$ and sum of the cost of all the window pairs in the subset approximates $\dyck(x)$. The input to the algorithm is string $x$, set of large and small windows $\mathcal{J}$ and $\mathcal{K}$ and parameters $\theta \in[0,1]$, $s_1,s_2,\delta \in [1\dd n]$ and $\alpha >0$ where $\alpha$ is a constant (we set these parameters in Section~\ref{sec:gapproof}). The algorithm progress in phases $i=0,1,2,4,\dots,10s_2$, where at the $i$th phase it uses a cost threshold $c_i=i$. In each phase the algorithm calls three procedures; first DeclareSparse(), followed by CertifyingSmall() and lastly CertifyingLarge() with parameter $c_i$.

Algorithm DeclareSparse() given a density parameter $\delta$ as input, for each window $w\in \mathcal{K}$ checks if there exists at least $\delta$ other windows $w'\in \mathcal{K}$ such that $\dyck(x[w]\circ x[w'])$ is at most $c_i$. If not then $w$ is declared $(c_i,\delta)$-sparse and added to set $S$. This is done using random sampling. Moreover, instead of directly computing the cost, we use Dyck-Approx() that provides $\alpha$ approximation of $\dyck(x[w]\circ x[w'])$.

Algorithm CertifyingSmall() starts by selecting a dense pivot window $w\in \mathcal{K}\setminus S$, and then it computes two set $A_1, A_2$, where $A_1$ contains all windows $w_1\in \mathcal{K}$ such that the approximated cost of  $x[w_1]\circ \overline{x[w]}$ is at most $2\alpha c$ and set $A_2$ contains all windows $w_2\in \mathcal{K}$ such that the approximated cost of  $x[w]\circ x[w_2]$ is at most $3\alpha^2 c$. Next for each pair of windows $(w_1,w_2)\in A_1\times A_2$ a weighted window pair $(w_1,w_2,5\alpha^2 c)$ is added to set $\mathcal{W}_S$. Apart from this for each close window pairs $(w,w')\in \mathcal{K}\times \mathcal{K}$, i.e. $s(w')-s(w)\le 5s_1$, CertifyingSmall() computes an approximation cost of $x[w]\circ x[w']$ using algorithm Dyck-Approx()
and adds the corresponding weighted window pair to $\mathcal{W}_S$.

The last algorithm CertifyingLarge() takes the set of sparse window $S$ as input and then for every large window $w\in \mathcal{J}$, it samples a set of sparse windows $w_1$ from $\mathcal{K}_w\cap S$. Next for each sampled sparse window $w_1$ it identifies all windows $w_2\in \mathcal{K}$ such that the approximated cost of $x[w_1]\circ x[w_2]$ is at most $\alpha c$. For each such $w_2$ it then identifies all the large windows $\tilde{w}\in \mathcal{J}$ that contains $w_2$ and computes an approximation of cost of $x[w]\circ x[\tilde{w}]$ using algorithm Dyck-Approx() and adds the corresponding weighted window pair to $\mathcal{W}_L$.

At the end Algorithm CertifyingWindowPairs() outputs all the weighted window pairs in set $\mathcal{W}_L$ and $\mathcal{W}_S$. The algorithm uses two more global parameters $k_1$ and $k_2$ that we set to be large constants.

%for any pair $(w,w')\in \mathcal{J}\times \mathcal{J}$ either it is certified by $\mathcal{W}_L$ or there exists a set of certified window pairs $\{(w_1,w'_1,c_1),\dots,((w_\ell,w'_\ell,c_\ell)\}\subseteq \mathcal{W}_S $ such that they form a consistent decomposition of $w\cup w'$ and sum of the cost of all these window pairs is at most constant times cost of $(w,w')$. 

\paragraph{$(c,\delta)$-dense and sparse.}

Given a window $w\in \mathcal{K}$ and integer parameters $c,\delta >0$, we call window $w$ to be $(c,\delta)$-dense if there exists at least $\delta$ windows $w'\in \mathcal{K}$ such that %$w'$ appears right of $w$ i.e., $s(w)<s(w')$ and
$\dyck(x[w]\circ x[w'])\le c$. Otherwise, we call $w$ $(c,\delta)$-sparse.

\iffalse
\paragraph{$(left,c,\delta)$-dense and sparse.}

Given a window $w\in \mathcal{K}$ and integer parameters $c,d >0$, we call window $w$ to be $(left,c,\delta)$-dense if there exists at least $\delta$ windows $w'\in \mathcal{K}$ such that $w'$ appears left of $w$ i.e., $s(w)>s(w')$ and $\dyck(x[w']\circ x[w])\le c$. Otherwise, we call $w$ $(left,c,\delta)$-sparse.

\paragraph{$(right,c,\delta)$-dense and sparse.}

Given a window $w\in \mathcal{K}$ and integer parameters $c,d >0$, we call window $w$ to be $(right,c,\delta)$-dense if there exists at least $\delta$ windows $w'\in \mathcal{K}$ such that $w'$ appears right of $w$ i.e., $s(w)<s(w')$ and $\dyck(x[w']\circ x[w])\le c$. Otherwise, we call $w$ $(right,c,\delta)$-sparse. 
\fi

\begin{algorithm}
		\KwIn{String $x$ of length $n$; set of windows $\mathcal{J},\mathcal{K}$; parameters $\theta\in [0,1]$, $\delta,s_1,s_2\in[n]$}
		
		\KwOut{A set $\mathcal{W}_L$ of certified window pairs in $\mathcal{J}\times \mathcal{J}$ and a set $\mathcal{W}_S$ of certified window pairs in $\mathcal{K}\times \mathcal{K}$}
		
		\ForEach{$i\in \{0,1,2,4,\dots,10s_2\}$}{
		 $c_i\gets i$\;
		
		$S\gets DeclareSparse(x,\mathcal{K},c_i,\theta,\alpha,\delta,s_1)$\;
		
         $\mathcal{W}_S\gets CertifiedSmall(x,\mathcal{K},S,c_i,\alpha,\delta,s_1)$\;
		
		 $\mathcal{W}_L\gets CertifiedLarge(x,\mathcal{J},\mathcal{K},S,c_i,\theta,\alpha,\delta,s_1,s_2)$\;

		}
		\Return{$\mathcal{W}_L$, $\mathcal{W}_S$}\;
		\caption{{CertifyingWindowPairs} $(x,\mathcal{J},\mathcal{K},\theta,\alpha,\delta,s_1,s_2)$}\label{alg:cover}
\end{algorithm}

\begin{algorithm}
		\KwIn{String $x$ of length $n$; set of windows $\mathcal{K}$; parameters $\theta \in [0,1]$, $c,\delta,s_1\in[n]$}
		\KwOut{Set $S$ of sparse windows in $\mathcal{K}$}
		$S \gets \emptyset$\;
		\ForEach{$w\in \mathcal{K}$}{		
		Sample $\frac{k_1|\mathcal{K}| \log n}{\delta}$ windows $w'\in \mathcal{K}$ uniformly at random and compute Dyck-Approx$(x[w]\circ x[w'])$\;
		\If{for at most $\frac{k_1}{2}\log n$ samples $w'$, Dyck-Approx$(x[w]\circ x[w'])\le \alpha c$}{
			Add $w$ to $S$\;
		}
		}
		\Return{$S$}\;
		\caption{{DeclareSparse} $(x,\mathcal{K},c,\alpha,\delta,s_2)$}\label{alg:checksparse}
\end{algorithm}

 \begin{algorithm}
		\KwIn{String $x$ of length $n$; set of windows $\mathcal{K},S$; parameters  $c,\delta,s_2\in[n]$}
		
		\KwOut{A set $\mathcal{W}_S$ of certified window pairs in $\mathcal{K}\times \mathcal{K}$}
			
		$\mathcal{W}_S\gets \emptyset$\;
		
		$\mathcal{L}\gets \mathcal{K}\setminus S$\;
		
		\While{$\mathcal{L}$ is nonempty}{
		
			Pick $w\in \mathcal{L}$\;
		
			$A_1\gets \{w_1\in \mathcal{K}$; Dyck-Approx$(x[w_1]\circ \overline{x[w]})\le 2\alpha c\}$\;
			$A_2\gets \{w_2\in \mathcal{K}$; Dyck-Approx$(x[w]\circ x[w_2])\le 3\alpha^2 c\}$\;
			Add $(w_1,w_2,5\alpha^2 c)$ to $\mathcal{W}_S$,  $\forall (w_1,w_2)\in A_1\times A_2$\;
			$\mathcal{L}\gets \mathcal{L}\setminus A_1$\;
		}
		
		\ForEach{$(w,w')\in \mathcal{K}\times \mathcal{K}$ such that $s(w')-s(w)\le 5s_1$}{
		
			$c\gets$ Dyck-Approx$(x[w]\circ x[w'])$\;
			Add $(w,w',c)$ to $\mathcal{W}_S$\;
		}
		
		\Return{$\mathcal{W}_S$}\;
		\caption{{CertifyingSmall} $(x,\mathcal{K},S,c,\alpha,\delta,s_2)$}\label{smallcertified}
\end{algorithm}

\begin{algorithm}
		\KwIn{String $x$ of length $n$; set of windows $\mathcal{J}$, $\mathcal{K}$, $S$; parameters $\epsilon,\theta \in [0,1]$, $c,\delta,s_1,s_2\in[n]$}
		
		\KwOut{A set $\mathcal{W}_L$ of certified window pairs in $\mathcal{J}\times \mathcal{J}$} 
		
		\ForEach{$w\in\mathcal{J}$}{
		
		Sample $\frac{k_2\log^2 n}{\theta^3}$ windows $w_1\in \mathcal{K}_w\cap S$ uniformly at random\;

		\ForEach{sampled $w_1$ and $w_2\in \mathcal{K}$}{
		
		\If{Dyck-Approx$(x[w_1]\circ x[w_2])\le \alpha c$}{
		
		\ForEach{$\tilde{w}\in \mathcal{J}$ such that $w_2\in \tilde{w}$}{
		
			$c'\gets$ Dyck-Approx$(x[w]\circ x[\tilde{w}])$\;
			
			Add $(w,\tilde{w},c')$ to $\mathcal{W}_L$\;

		}

		}
		
		}

		}

		\Return{$\mathcal{W}_L$}
		\caption{{CertifyingLarge} $(x,\mathcal{J},\mathcal{K},S,S,c,\theta,\alpha,\delta,s_1,s_2)$}\label{largecertified}
\end{algorithm}

\subsection{Correctness}

\begin{theorem}\label{thm:maincorrect}
Let $n$ be a sufficiently large power of $2$ and $\theta \in[\frac{1}{n},1]$ be a power of $2$ and $\epsilon>0$ a constant. Let $x$ be a string of length $n$. Let $s_1,s_2\in [\frac{1}{\theta}, n]$, $\delta \in[n]$ be a power of $2$ where $s_2|s_1$ and $s_1|n$ and $\frac{s_2}{s_1}\le \theta$.
Let $\mathcal{W}_L$ and $\mathcal{W}_S$ be the set of weighted window pairs generated by algorithm CertifyingWindowPairs. We claim every weighted window pair in $\mathcal{W}_L$ and $\mathcal{W}_S$ is correctly certified. With high probability,
there exists a consistent decomposition of $[1\dd n]$ into weighted window pairs $\{(w_1,w'_1,c_1),\dots,(w_\ell,w'_\ell,c_\ell)\}\subseteq \mathcal{W}_S\cup \W_L$ such that $\sum_{j=1}^\ell c_j \le 10 \alpha^2 \dyck(x)+O(\theta n)$.
\end{theorem}

We start with few definitions required to prove \cref{thm:maincorrect}. Consider a window pair $(w,w')\in \mathcal{J}\times \mathcal{J}$. Fix an optimal alignment $\tilde{M}$ of $x[w]\circ x[w']$. Let $\mathcal{\tilde{S}}_{\tilde{M}}^{(w,w')}$ be the set generated by applying the process of \emph{capped small alignment window decomposition} on the interval $[w\cup w']$ with respect to alignment $\tilde{M}$. Notice $\mathcal{\tilde{S}}_{\tilde{M}}^{(w,w')}\subseteq \mathcal{K}\times \mathcal{K}$ and by  construction for any window pair $(q,q')\in \mathcal{\tilde{S}}_{\tilde{M}}^{(w,w')}$, $|q|+|q'|\le 10s_2$. Thus, $\dyck(x[q]\circ x[q'])\le 10s_2$. 

\paragraph{Set $B_c^{(w,w')}$.}
For each $c\in \{0,1,2,4,\dots,10s_2\}$, we define set $B_c^{(w,w')}$ containing all the window pairs $(q,q')\in \mathcal{\tilde{S}}_{\tilde{M}}^{(w,w')}$ such that 
$\frac{c}{2} <\dyck(x[q]\circ x[q'])\le c$.

\noindent
\textbf{Set $B_{c,d}^{(w,w')}$ and $B_{c,s}^{(w,w')}$.}
We further partition set $B_c^{(w,w')}$ into two groups. Define set $B_{c,d}^{(w,w')}$ containing all the window pairs $(q,q')\in B_c^{(w,w')}$ such that $q$ is $(c,\delta)$-dense. %or $q'$ is is $(left,c,\delta)$-dense. 
Otherwise, we add $(q,q')$ to $B_{c,s}^{(w,w')}$.

\noindent
\textbf{Set $B_{c,s,w}^{(w,w')}$, $B_{c,s,w'}^{(w,w')}$ and $B_{c,s,(w,w')}^{(w,w')}$.}
Finally, we partition set $B_{c,s}^{(w,w')}$ into three groups. Define set $B_{c,s,w}^{(w,w')}$ containing all the window pairs $(q,q')\in B_{c,s}^{(w,w')}$ such that both $q,q'\subseteq w$.
Define set $B_{c,s,w'}^{(w,w')}$ containing all the window pairs $(q,q')\in B_{c,s}^{(w,w')}$ such that both $q,q'\subseteq w'$. Lastly define set $B_{c,s,(w,w')}^{(w,w')}$ containing all the window pairs $(q,q')\in B_{c,s}^{(w,w')}$ such that $q\subseteq w$ and $q'\subseteq w'$.

\paragraph{Favorable Sampling.} The algorithm uses random sampling twice. First in the \emph{DeclareSparse} algorithm to identify whether a window $w\in \mathcal{K}$ is dense or sparse and secondly in the algorithm \emph{CertifyingLarge} where for each $w\in \mathcal{J}$, we sample sparse windows from $\mathcal{K}_w\cap S$. We now define when do we admit these samplings to be favorable.

\begin{definition}
During a run of CertifyingWindowPairs for each $i=0,1,2,4,\dots, 10s_2$, of the FOR loop we call the samplings to be favorable if they satisfy the following.
\begin{itemize}
    \item In algorithm DeclareSparse for each $w\in \mathcal{K}$ if $w$ is $(c_i,\delta/4)-sparse$ then $w\in S$ and otherwise if $w$ is $(c_i,\delta)-dense$ then $w\notin S$.
    %Similarly, if $w$ is $(left,c_i,\delta/4)-sparse$ then $w\in S_L$ and otherwise if $w$ is $(right,c_i,\delta)-dense$ then $w\notin S_L$.
    \item In the algorithm CertifyingLarge for every window pair $(w,w')\in \mathcal{J}\times \mathcal{J}$, 
    %if $|B_{c_i,w}^s|\ge \frac{\theta |w|}{s_2 \log n}$ then the algorithm samples a window $q\in \mathcal{K}_w\cap S_R$, such that $\exists q'\in \mathcal{K}_{w}$ and $(q,q')\in B_{c_i,w}^s$. 
    %Further if $|B_{c_i,w'}^s|\ge \frac{\theta |w|}{s_2 \log n}$ then the algorithm samples a window $q\in \mathcal{K}_{w'}\cap S_R$, such that $\exists q'\in \mathcal{K}_{w'}$ and $(q,q')\in B_{c_i,w'}^s$. 
    $|B_{c,s,(w,w')}^{(w,w')}|\ge \frac{\theta |w|}{s_2 \log n}$ then the algorithm samples a window $q\in \mathcal{K}_w\cap S$, such that $\exists q'\in \mathcal{K}_{w'}$ and $(q,q')\in B_{c,s,(w,w')}^{(w,w')}$.
\end{itemize}
\end{definition}

\begin{lemma}\label{lem:favsamp}
For large $n$, in any run of CertifyingWindowPairs the samplings are favorable with probability at least $1-\frac{1}{n^9}$.
\end{lemma}

\begin{proof}
In the \emph{DeclareSparse} algorithm consider a $w\in \mathcal{K}$ and some $c_i$. If $w$ is $(c_i,\delta)-dense$ then there exists at least $\delta$ other windows $w'\in \mathcal{K}$, such that $\dyck(x[w]\circ x[w'])\le c_i$. Let $D_R$ be this set of windows. As we sample $\frac{k_1|\mathcal{K}|\log n}{\delta}$ samples from right of $w$, for large constant $k_1$, using Chernoff bound probability that less than $\frac{k_1 \log n}{2}$ windows are sampled from $D_R$ is at most $1/n^{15}$. On the other hand if $w$ is $(c_i,\delta/4)-sparse$ then size of $D_R$ is at most $\delta/4$. Thus, probability that at least $\frac{k_1 \log n}{2}$ windows are sampled from $D_R$ is at most $1/n^{15}$. %Similarly, we can claim for $S_L$.
As there are $\log n$ choices for $i$ and $\tOh(\frac{n}{\theta s_2})$ choices for $w\in \mathcal{K}$ the first criteria is not satisfied with probability at most  $1/n^{10}$. 

In the \emph{CertifyingLarge} algorithm consider a window pair $(w,w')\in \mathcal{J}\times \mathcal{J}$ and some $c_i$. Notice $\mathcal{K}_w=O(\frac{|w|}{\theta^2 s_2})$. Moreover, as $|B_{c_i,s,(w,w')}^{(w,w')}|\ge \frac{\theta |w|}{s_2 \log n}$, if we sample $\frac{k_2 \log^2 n}{\theta^3}$ windows from $\mathcal{K}_w$, then for large constant $k_2$, using Chernoff bound probability that less than $\frac{k_2 \log n}{2}$ windows $q$ are sampled such that if $q'$ is the matching window of $q$ (w.r.t. $\tilde{M}$), then $(q,q')\notin B_{c_i,s,(w,w')}^{(w,w')}$ is at most $1/n^{15}$. %Similarly, we can argue for set $B_{c_i,w'}^s$ and set $B_{c_i,(w,w')}^s$. 
As there are $\log n$ choices for $i$ and $\tOh(\frac{n^2}{\theta^2 s_1^2})$ choices for $(w,w')\in \mathcal{J}\times \mathcal{J}$ the second criteria is not satisfied with probability at most  $1/n^{10}$.
Thus, probability that the samplings are not favorable is at most $1/n^9$.

\end{proof}

For every window pair $(w,w')\in \mathcal{J}\times \mathcal{J}$, each optimal alignment $\tilde{M}$ of $x[w]\circ x[w']$ and each cost parameter $c\in\{0,1,2,4,\dots,10s_2\}$ we can make the following claim.

\begin{claim}\label{claim:densecorrect}
For each $c\in\{0,1,2,4,\dots,10s_2\}$, for each pair of windows $(q,q')\in B_c\setminus B_{c,s,(w,w')}^{(w,w')}$, there exists a weighted window pair $(q,q',c')\in \mathcal{W}_S$ such that $c'\le 10\alpha^2 \dyck(x[q]\circ x[q'])+O(1)$.
\end{claim}

\begin{proof}

Notice $B_c\setminus B_{c,s,(w,w')}^{(w,w')}=B_{c,d}^{(w,w')}\cup B_{c,s,w}^{(w,w')}\cup B_{c,s,w'}^{(w,w')}$. Let us first consider the case where $(q,q')\in B_{c,s,w}^{(w,w')}$. As in this case both $s(q),s(q')\in w$, $s(q')-s(q)\le 5s_1$. Thus, in \cref{smallcertified}, Step 11, we call Dyck-Approx() on string $x[q]\circ x[q']$ which computes $\alpha$ approximation of $\dyck(x[q]\circ x[q'])$. Thus, \cref{smallcertified} adds a weighted window pair $(q,q',c')$ where $c'\le \alpha \dyck(x[q]\circ x[q'])$, and thus we prove our claim. Following a similar argument we can prove the case when $(q,q')\in B_{c,s,w'}^{(w,w')}$.

Next consider the case where $(q,q')\in B_{c,d}^{(w,w')}$.
For this we show that there exists a weighted window pair $(q,q',c')\in \mathcal{W}_S$ such that $c'\le 5\alpha^2 c$. Now as by definition $c/2< \dyck(x[q]\circ x[q'])\le c$, the claim follows.
By definition here $q$ is $(c,\delta)$-dense. Then assuming favorable sampling $w\notin S$. 
Thus, in \cref{smallcertified} either at Step 4 we select $q$ as a pivot and in this case $q\in A_1$ and $q'\in A_2$. Thus, \cref{smallcertified} adds a weighted window pair $(q,q',c')$ to $\mathcal{W}_S$ where $c'=5\alpha^2 c$, and thus we prove our claim. Otherwise, if $q$ is not selected as a pivot then there exists a pivot $\tilde{q}$, such that Dyck-Approx$(x[q]\circ \overline{x[\tilde{q}]})\le 2\alpha c$ and thus $q\in A_1$. But this implies $dist(x[q]\circ \overline{x[\tilde{q}]})\le 2\alpha c$. Now as $dist(x[q]\circ x[q'])\le c$, by triangle inequality we have $dist(x[\tilde{q}]\circ x[q'])\le 2\alpha c+c \le 3 \alpha c$. Thus, Dyck-Approx$(x[\tilde{q}]\circ x[q'])\le 3 \alpha^2 c$ and $q'\in A_2$. Hence, again \cref{smallcertified} adds a weighted window pair $(q,q',c')$ to $\mathcal{W}_S$ where $c'=5\alpha^2 c$, and thus we prove our claim.
\end{proof}

\begin{proposition}\label{prp:maincorrect}
Let $n$ be a sufficiently large power of $2$ and $\theta \in[\frac{1}{n},1]$ be a power of $2$ and $\epsilon>0$ a constant. Let $x$ be a string of length $n$. Let $s_1,s_2\in [\frac{1}{\theta}, n]$, $\delta \in[n]$ be a power of $2$ where $s_2|s_1$ and $s_1|n$
and $\frac{s_2}{s_1}\le \theta$. 
Let $\mathcal{W}_L$ and $\mathcal{W}_S$ be the set of weighted window pairs generated by algorithm CertifyingWindowPairs. We claim every weighted window pair in $\mathcal{W}_L$ and $\mathcal{W}_S$ is correctly certified. With high probability, for every window pair $(w,w')\in \mathcal{J}\times \mathcal{J}$, at least one of the following is satisfied:
\begin{itemize}
    \item there exists a weighted window pair $(w,w',c')\in \mathcal{W}_L$, such that $c'\le \alpha \dyck(x[w]\circ x[w'])$, or
    \item there exists a consistent decomposition of $w\cup w'$ into weighted window pairs $T^{(w,w')}=\{(w_1,w'_1,c_1),\allowbreak \dots,(w_\ell,w'_\ell,c_\ell)\}\subseteq \mathcal{W}_S$ such that $\sum_{j=1}^\ell c_j \le 10 \alpha^2 \dyck(x[w]\circ x[w'])+O(\theta s_1)$.
\end{itemize}
\end{proposition}

\begin{proof}
First we claim that each weighted window pair in $\mathcal{W}_L$ and $\mathcal{W}_S$ is correctly certified. For each weighted window pair in $\mathcal{W}_L$ the cost is estimated by algorithm Dyck-Approx and thus is at least the actual cost of the corresponding window pairs. Next for each weighted window pair $(q,q',c')\in \mathcal{W}_S$ there exists $\tilde{q}$ such that $\dyck(x[q]\circ \overline{x[\tilde{q}]})\le c_1$ and $\dyck(x[\tilde{q}]\circ x[q'])\le c_2$ where $c_1+c_2=c$. Thus, by triangle inequality $\dyck(x[q]\circ x[q'])\le c$.

Consider a pair of windows $(w,w')\in \mathcal{J}\times \mathcal{J}$. Fix an optimal alignment $\tilde{M}$ of $x[w]\circ x[w']$. Let $\mathcal{\tilde{S}}_{\tilde{M}}^{(w,w')}$ be the set generated by applying the process of \emph{capped small alignment window decomposition} on the interval $[w\cup w']$ with respect to alignment $\tilde{M}$. We divide our analysis into two cases. First consider the case where $\exists c$ such that $|B_{c,s,(w,w')}^{(w,w')}|\ge \frac{\theta |w|}{s_2\log n}$. In this case assuming favorable sampling \cref{largecertified} samples a window $q\in \mathcal{K}_w\cap S$, such that $\exists q'\in \mathcal{K}_{w'}$ where $(q,q')\in B_{c,s,(w,w')}^{(w,w')}$. But this implies $\dyck(x[q]\circ x[q'])\le c$ and thus algorithm Dyck-Approx certifies that $\dyck(x[q]\circ x[q'])\le \alpha c$. Hence, \cref{largecertified} calls Dyck-Approx on $x[w]\circ x[w']$ which provides an estimation $c'$ where $c'\le \alpha \dyck(x[w]\circ x[w'])$.

Next consider the case where $\forall c$, $|B_{c,s,(w,w')}^{(w,w')}|< \frac{\theta |w|}{s_2\log n}$. 
Notice for $c=10s_2$, if we consider any window pair $(q,q')\in \mathcal{\tilde{S}}_{\tilde{M}}^{(w,w')}$, then $q$ is $(c,\delta)$-dense. Thus, following this we can argue that for each window pair $(q,q')\in B_{c,s,(w,w')}^{(w,w')}$, there exists a weighted window pair $(q,q',c')\in \mathcal{W}_S$ such that $c'\le 5\alpha^2c=50\alpha^2 s_2$. As $|B_{c,s,(w,w')}^{(w,w')}|< \frac{\theta |w|}{s_2\log n}$ and 
 there are $O(\log n)$ different values of $c$, total estimated cost of all windows in $B_{c,s,(w,w')}^{(w,w')}$ is at most $\frac{\theta |w|}{s_2\log n}\cdot \log n \cdot 50\alpha^2 s_2=O(\theta s_1)$ (as $|w|\le 5s_1$). 
 Formally we claim 
 \begin{equation*}
   \sum_{c}\sum_{\substack{(q,q')\in B_{c,s,(w,w')}^{(w,w')} \\ (q,q',c')\in \mathcal{W}_S}} c'=O(\theta s_1)  
 \end{equation*}

 On the other hand, following \cref{claim:densecorrect}, for each $c$ we can ensure that, for each window pair $(q,q')\in B_c\setminus B_{c,s,(w,w')}^{(w,w')}$, there exists a weighted window pair $(q,q',c')\in \mathcal{W}_S$ such that $c'\le 10\alpha^2 \dyck(x[q]\circ x[q'])+O(1)$. By definition for any two $c_i,c_j\in \{0,1,2,\dots,10s_2\}$, if $c_i\neq c_j$ then $B_{c_i}\cap B_{c_j}=\emptyset$ and for each $c_i$, $B_{c_i}\subseteq \mathcal{\tilde{S}}_{\tilde{M}}^{(w,w')}$. Thus, 

\begin{align*}
   \sum_{c}\sum_{\substack{(q,q')\in B_c\setminus B_{c,s,(w,w')}^{(w,w')} \\ (q,q',c')\in \mathcal{W}_S}} c'&\le \sum_{c}\sum_{(q,q')\in B_c\setminus B_{c,s,(w,w')}^{(w,w')}} (10\alpha^2 \dyck(x[q]\circ x[q'])+O(1))\\
   & \le \sum_{(q,q')\in \mathcal{\tilde{S}}_{\tilde{M}}^{(w,w')}} (10\alpha^2 \dyck(x[q]\circ x[q'])+O(1)) \\
    & \le 10 \alpha^2\sum_{(q,q')\in \mathcal{\tilde{S}}_{\tilde{M}}^{(w,w')}} \dyck(x[q]\circ x[q'])+O(|\mathcal{\tilde{S}}_{\tilde{M}}^{(w,w')}|)\\
 & \le 10 \alpha ^2 \cost_{\tilde{M}}(x[w]\circ x[w'])+O(\theta s_1) \\
 & = 10 \alpha ^2 \dyck(x[w]\circ x[w'])+O(\theta s_1)
 \end{align*}
 Moreover, as $\mathcal{\tilde{S}}_{\tilde{M}}^{(w,w')}=\cup_{c\in\{0,1,2,4,\dots,8s_2\}}B_c^{(w,w')}$, together with above two equations
  we can claim that for each pair of windows in $(q,q')\in \mathcal{\tilde{S}}_{\tilde{M}}^{(w,w')}$, there is a tuple $(q,q',c')\in \mathcal{W}_S$ such that the sum of the cost of all the tuples is at most $10 \alpha ^2 \dyck(x[w]\circ x[w'])+O(\theta s_1)$.
 
Thus, we set $T^{(w,w')}=\{(q,q',c'); (q,q')\in \mathcal{\tilde{S}}_{\tilde{M}}^{(w,w')}, (q,q',c')\in \mathcal{W}_S\}$. Notice by definition of $\mathcal{\tilde{S}}_{\tilde{M}}^{(w,w')}$, the window pairs in $T$ induces a consistent decomposition of $w\cup w'$. This completes the proof of Proposition~\ref{prp:maincorrect}.
\end{proof}

\begin{proof}[Proof of Theorem~\ref{thm:maincorrect}]
Fix an optimal alignment $M$ of $x$. Let $\mathcal{\tilde{S}}_M$ be the set generated by applying the process of capped large alignment window decomposition on the interval $[1\dd n]$.
Notice by construction the window pairs of $\mathcal{\tilde{S}}_M$ provides a consistent decomposition of $[1\dd n]$.
For each window pair $(w,w')\in \mathcal{\tilde{S}}_M$ by Proposition~\ref{prp:maincorrect} either there exists a certified tuple $(,w',c)\in \mathcal{W}_L$ such that $c\le \alpha \dyck(x[w]\circ x[w'])$ or here exists a consistent decomposition of $w\cup w'$ into weighted window pairs $T^{(w,w')}=\{(w_1,w'_1,c_1),\allowbreak \dots,(w_\ell,w'_\ell,c_\ell)\}\subseteq \mathcal{W}_S$ such that $\sum_{j=1}^\ell c_j \le 10 \alpha^2 \dyck(x[w]\circ x[w'])+O(\theta s_1)$. Combining this with Lemma~\ref{lem:cappedlargedecomposition} that says $\sum_{(w,w')\in \mathcal{\tilde{S}}_M} \dyck(x[w]\circ x[w'])\le cost_M(x)+O(\theta n)$ we prove Theorem~\ref{thm:maincorrect}. Note as $M$ is an optimal alignment of $x$, $\dyck(x)=\cost_M(x)$.
\end{proof}

\subsection{Running Time Analysis}

Let $t(|w_1|+|w_2|)$ be the running time of the Dyck-Approx$(x[w_1]\circ x[w_2])$ that computes $\alpha$ approximation of Dyck edit distance of a string of length $(|w_1|+|w_2|)$. Here $\alpha\ge 1$ is a constant. 

\begin{theorem}\label{thm:maintime}
Let $n$ be a sufficiently large power of $2$ and $\theta \in[\frac{1}{n},1]$ be a power of $2$ and $\epsilon>0$ a constant. Let $x$ be a string of length $n$. Let $s_1,s_2\in [\frac{1}{\theta}\dd n]$, $\delta \in[n]$ be a power of $2$ where $s_2|s_1$ $s_1|n$
and $\frac{s_2}{s_1}\le \theta$. 
\begin{enumerate}
\item The size of set $(\mathcal{W}_L\cup \mathcal{W}_S)$ output by \emph{CertifyingWindowPairs} algorithm is $\tOh(\frac{n^2}{\theta^4 s_2^2})$.
\item Any run that satisfies favorable sampling, \emph{CertifyingWindowPairs} runs in time
\[\tOh\left((|\mathcal{W}_L|+|\mathcal{W}_S|)+\tfrac{n^2}{\theta ^4 \delta s_2^2}\cdot t(10s_2)+ \tfrac{ns_1}{\theta^4 s_2^2}\cdot t(10s_2)+\tfrac{n^2}{\theta^7 s_1s_2}\cdot t(10s_2)+\tfrac{n\delta}{\theta^7 s_1}\cdot t(10s_1)\right)\]
\end{enumerate}

\end{theorem}

\begin{proof}
 To bound $|\mathcal{W}_S|$ notice by \cref{claim:sizesmall} $|\mathcal{K}|=\tOh(\frac{n}{\theta^2 s_2})$. As $\mathcal{W}_S\subseteq \mathcal{K}\times \mathcal{K}$, we get $|\mathcal{W}_S|=\tOh(\frac{n^2}{\theta^4 s_2^2})$. Similarly, by \cref{claim:sizelarge} $|\mathcal{J}|=\tOh(\frac{n}{\theta^2 s_1})$. As $\mathcal{W}_L\subseteq \mathcal{J}\times \mathcal{J}$, $|\mathcal{W}_L|=\tOh(\frac{n^2}{\theta^4 s_1^2})$. As $s_1>s_2$ we get the required bound.
 
 Next we analyze the running time. The \emph{CertifyingWindowPairs} algorithm outputs a certified window pair in $O(1)$ time, thus the total time required is $O(|\mathcal{W}_L|+|\mathcal{W}_S|)$. Next the algorithm 
 calls three subroutines, \emph{DeclareSparse}, \emph{CertifyingLarge} and \emph{CertifyingSmall} with $O(\log 10s_2)=O(\log n)$ different values of $c_i$. We next compute the running taken of each subroutine for a single call.
 
 We start by computing the running time of \emph{DeclareSparse} that a set of sparse windows $S$. To compute $S$, for each window $w\in \mathcal{K}$ we sample $\tOh(\frac{|\mathcal{K}|}{\delta})$ (here $k_1>0$ is a constant) windows $w'$ from $\mathcal{K}$ and for each pair $(w,w')$ we run the Dyck-Approx algorithm on string $x[w]\circ x[w']$ that takes time $t(|w|+|w'|)\le t(10s_2)$ as $|w|,|w'|\le 5s_2$. For each $w\in \mathcal{K}$ we call Dyck-Approx $\tOh(\frac{|\mathcal{K}|}{\delta})$ times, thus the time required to process each window is $\tOh(\frac{|\mathcal{K}|}{\delta }\cdot t(10s_2))$. Hence, all the windows in $\mathcal{K}$ can be processed in time $\tOh(\frac{|\mathcal{K}|^2}{\delta }\cdot t(10s_2))$. %Similarly, we can construct set $S_L$ is time $\tOh(\frac{|\mathcal{K}|^2}{\delta \theta}\cdot t(10s_2))$.
 As $|\mathcal{K}|=\tOh(\frac{n}{\theta^2 s_2})$ and $|S|\le |\mathcal{K}|$, \emph{DeclareSparse} runs in time $\tOh(\frac{n^2}{\theta^4 \delta s_2^2}\cdot t(10s_2))$.
 
 Next we analyze \emph{CertifyingSmall}. Let us first compute the time required to process one pivot window $w\in \mathcal{L}=\mathcal{K}\setminus S$. For this to compute set $A_1$ consider each window $w_1\in \mathcal{K}$ and call Dyck-Approx algorithm with input $x[w_1]\circ \overline{x[w]}$ that runs in time $t(|w_1|+|w|)\le t(10s_2)$. Thus, total time required to construct set $A_1$ for one pivot is $|\mathcal{K}|\cdot t(10s_2)$. Similarly, we can construct set $A_2$ in time $|\mathcal{K}|\cdot t(10s_2)$. Thus, total time required to process one pivot window of $\mathcal{L}$ is $\tOh(|\mathcal{K}|\cdot t(10s_2))$. Next we bound the number of windows that are selected as pivot before we process all the windows in $\mathcal{L}$. %For this we partition $\mathcal{K}\setminus S_R$ according to the size of the windows while keeping all same size windows in one group. Notice as in $\mathcal{K}$ there are at most $\frac{4}{\theta}$ different sized windows total number of groups is at most $\frac{4}{\theta}$. We analyse the running time for each group separately. 
 Notice after processing a pivot window $w$ we remove set $A_1$ from $\mathcal{L}$. Thus, if $w,w'$ be two different pivots selected then Dyck-Approx$(x[w']\circ \overline{x[w]})>2\alpha c$.
 But this implies $\dyck(x[w']\circ \overline{x[w]})>2c$.
 Thus, the set $\mathcal{B}_R(w)=\{\tilde{w}\in \mathcal{K}; \dyck(x[w]\circ x[\tilde{w}])\le c\}$ are disjoint for $w$ and $\tilde{w}$. By favorable sampling $|\mathcal{B}_R(w)|\ge \frac{\delta}{4}$. Thus, total number of iterations required to process all windows is $\tOh(\frac{4|\mathcal{K}|}{\delta})$. Hence, we can process all windows in $\mathcal{L}$ in time $\tOh(\frac{|\mathcal{K}|^2}{\delta}\cdot t(10s_2))=\tOh(\frac{n^2}{\theta^4 \delta s_2^2}\cdot t(10s_2))$.
  %Similarly, we can process set $\mathcal{K}\setminus S_L$ in time $\tOh(\frac{n^2}{\theta^2 \delta s_2^2}\cdot t(8s_2))$.
  Next for each $w\in \mathcal{K}$, the algorithm considers all windows $w'\in \mathcal{K}$ such that $s(w')-s(w)\le 5s_1$ and calls Dyck-Approx on $x([w]\circ [w'])$. Each such call takes time $t(10s_2)$. As for each $w$ there are $\tOh(\frac{s_1}{\theta^2 s_2})$ different choices for $w'$, total time required is $\tOh(\frac{ns_1}{\theta^4 s_2^2}\cdot t(10s_2))$.
 
 Lastly we analyze the running time of \emph{CertifyingLArge}. To process a single window $w\in \mathcal{J}$ we start by sampling $\frac{k_2\log^2 n}{\theta^3}$ windows $w_1$ from $\mathcal{K}_w\cap S$ (where $k_2>0$ is a constant). For each such sampled window $w_1$ we identify all windows $w_2\in \mathcal{K}$ such that Dyck-Approx$(x[w_1]\circ x[w_2])\le \alpha c$. Each such checking can be done in time $t(|w_1|+|w_2|)=t(10s_2)$ as $|w_1|,|w_2|\le 5s_2$. Hence, total time required to identify all windows $w_2$ at a small distance from $w_1$ is $|\mathcal{K}|\cdot t(10s_2)$. Moreover, assuming favorable sampling as $w_1\in \mathcal{K}_w\cap S$, there exists at most $\delta$ such $w_2$ at small distance. Hence, for all sampled $\frac{k_2 \log^2 n}{\theta^3}$ windows $w_1$ we can identify all close windows $w_2$ in time $\tOh(\frac{|\mathcal{K}|}{\theta^3}\cdot t(10s_2))$. Next for a sampled window $w_1$ and a window $w_2$ close to it, we consider each $\tilde{w}\in \mathcal{J}$ such that $w_2\subseteq \tilde{w}$. Notice there are at most $\tOh(\frac{1}{\theta^2})$ such $\tilde{w}$. For each $\tilde{w}$, we again call the Dyck-Approx algorithm with input $x[w]\circ x[\tilde{w}]$ that takes time $t(|w|+|\tilde{w}|)=t(10s_1)$ as $|w|,|\tilde{w}|\le 5s_1$. Hence, we can check for all $\tilde{w}$ in time $\tOh(\frac{1}{\theta^2}\cdot t(10s_1))$. Thus, for a fixed sampled $w_1$ we can process all $\delta$ close windows $w_2$ in time $\tOh(\frac{\delta}{\theta^2}\cdot t(10s_1))$. Total time required to compute the extension of all $\frac{k_2 \log^2 n}{\theta^3}$ sampled windows $\tOh(\frac{\delta}{\theta^5}\cdot t(10s_1))$. Thus, we can process all $\frac{k_2\log^2 n}{\theta^3}$ sampled windows $w_1\in \mathcal{K}_w\cap S$ in time $\tOh(\frac{\delta}{\theta^5}\cdot t(10s_1)+\frac{|\mathcal{K}|}{\theta^3}\cdot t(10s_2))$. Thus, the total running time to process a single window $w\in \mathcal{J}$ is $\tOh(\frac{\delta}{\theta^5}\cdot t(10s_1)+\frac{|\mathcal{K}|}{\theta^3}\cdot t(10s_2))$. Hence, total running time of \emph{CertifyingLarge} is $\tOh(|\mathcal{J}|(\frac{\delta}{\theta^5}\cdot t(10s_1)+\frac{|\mathcal{K}|}{\theta^3}\cdot t(10s_2)))= \tOh(\frac{n\delta }{\theta^7 s_1}\cdot t(10s_1)+\frac{n^2}{\theta^7 s_1 s_2}\cdot t(10s_2))$.
\end{proof}

\subsection{Dynamic Program}

\begin{proposition}
	Given an integer $n\ge 1$ and a set $\W$ of weighted window pairs,
	one can minimize the total weight $\sum_{j=1}^\ell c_j$ over all consistent decompositions of $[1\dd n]$ into weighted window pairs $\{(w_1,w'_1,c_1),\ldots,(w_\ell,w'_\ell,c_\ell)\}\subseteq \W$
	in $\tOh(|C|^3+ |\W|)$ time, where $C = \bigcup_{(w,w',c)\in \W} \{s(w)-1,e(w),s(w')-1,s(w')\}$.
\end{proposition}

\begin{algorithm}[htb]
	$C := \bigcup_{(w,w',c)\in \W} \{s(w)-1,e(w),s(w')-1,s(w')\}$\;
	\For{$i \in C$ (in the decreasing order)}{
		$D[i,i] := 0$\;
		\For{$j := (i\dd n]\cap C$ (in the increasing order)}{
			$D[i,j] := \infty$\;
			\ForEach{$(w,w',c)\in \W$ such that $s(w)-1=i$ and $e(w')=j$}{
				$D[i,j] := (D[i,j], c + D[e(w),s(w')-1])$\;
			}
			\For{$k \in (i\dd j)\cap C$}{
				$D[i,j] := \min(D[i,j], D[i,k]+D[k,j])$\;
			}
		}
	}
	\Return{$D[0,n]$}\;

	\caption{The DP computing the minimum-weight consistent decomposition of $[1\dd n]$.}\label{algo:dp}
\end{algorithm}

\begin{proof}
The algorithm computes a DP table $D[\cdot,\cdot]$
such that $D[i,j]$ is the minimum weight of a consistent decomposition of $(i\dd j]$.
Consider the optimum solution $T\subseteq \W$
and a tuple $(w,w',c)\in T$ such that $s(w)=i+1$ (in case of many options,
choose one maximizing $e(w')$).
If $e(w')=j$, then $T$ is considered in the \textbf{foreach} loop.
Otherwise, by \cref{def:consistent}, all the tuples
have both windows inside $(i\dd e(w')]$ or inside $(e(w')\dd j]$.
Thus, $T$ is considered in the innermost \textbf{for} loop with $k=e(w')$.

As far as the implementation is concerned, we spend $\tOh(|\W|)$ time computing $C$ 
and mapping each endpoint to its rank in $C$.
The remaining running time is $\Oh(|C|^3+|\W|)$.
\end{proof}

The instance arising from \cref{thm:maincorrect,thm:maintime} satisfies $|\W|=\tOh(\frac{n^2}{\theta^4s_2^2})$,
and $|C|=\Oh(\frac{n}{\theta s_2})$ (recall that each window in $\mathcal{J}$ and $\mathcal{K}$ ends at a multiple of $\theta s_2$ and has its length divisible by $\theta s_1$). Consequently, the running time of \cref{algo:dp}
on this instance is $\tOh(\frac{n^3}{\theta^3 s_2^3} + \frac{n^2}{\theta^4s_2^2})$.

\subsection{Proof of Theorem~\ref{thm:gapapproxsubquadratic}}\label{sec:gapproof}
Given a string $x$ of length $n$ over alphabet $\Sigma$ and a parameter $\theta\in [n^{-1/34},1]$ we first create sets of large and small windows $\mathcal{J}, \mathcal{K}$ as described in Section~\ref{sec:window}. Next we run algorithm CertifyingWindowPairs(Algorithm~\ref{alg:cover}) with input string $x$, set of windows $\mathcal{J},\mathcal{K}$ parameters $\theta\in [n^{-1/34},1]$, $\delta, s_1, s_2, \alpha$ that generates a set $\mathcal{W}=\mathcal{W}_L\cup \mathcal{W}_S$ of weighted window pairs such that there exists a subset $T=\{(w_1,w'_1,c_1),\dots,(w_\ell,w'_\ell,c_\ell)\}\subseteq \mathcal{W}$ that forms a consistent decomposition of $[1\dd n]$ and $\sum_{i=1}^\ell c_i \le 10\alpha^2 \dyck(x)+O(\theta n)$ (Theorem~\ref{thm:maincorrect}). Next we run the dynamic program Algorithm~\ref{algo:dp} with input $\mathcal{W}$ that minimizes $\sum_{i=1}^\ell c_i$ over weighted decomposition $T$ and thus provides a cost estimation $10\alpha^2 \dyck(x)+O(\theta n)$. 
Here $\alpha$ is the approximation factor of algorithm Dyck-Approx(Algorithm~\ref{alg:Dyck-Approx}). By Theorem~\ref{thm:approxdyck} setting $\epsilon=1$, we get $\alpha=2$.
Thus, if $\dyck(x)\le \theta n$ we get an estimation $O(\theta n)$. 
Note here $t(|w|)$ is the running time Algorithm~\ref{alg:Dyck-Approx} with input string $w$; thus $t(w)=\tOh(w^2)$. By setting  $s_1,s_2,\delta$ to be the largest power of $2$ satisfying $s_1\le n^{21/34}$, $s_2\le n^{13/34}$ and $\delta\le n^{5/34}$ we get the running time of Algorithm~\ref{alg:cover} and Algorithm~\ref{algo:dp} is $n^{1.971}$. This completes the proof of Theorem~\ref{thm:gapapproxsubquadratic}.

\section{RNA Folding}\label{sec:fold}
\newcommand{\RNA}{\mathsf{FOLD}}
\newcommand{\rna}{\mathsf{fold}}

The alphabet $\Sigma$ is associated with a fixed-point free involution $\rev{\cdot} : \Sigma \to \rev{\Sigma}$ mapping each symbol to its \emph{complement}.
The folding language $\RNA(\Sigma)$ can be defined using a context-free grammar whose only non-terminal $S$ admits productions $S\to SS$, $S\to \emptyset$ (empty string), and $S\to aS\rev{a}$ for $a\in \Sigma$.

\begin{definition}
	The folding edit distance $\rna(x)$ of a string $x\in \Sigma^*$ is the minimum number of character deletions required to transform $x$ to a string in $\RNA(\Sigma)$.
\end{definition}

In this section we provide an algorithm that computes $\tau$ approximation of $\rna(x)$ in time $O(n+n^3/\tau^3)$. Formally, we show the following.

\begin{restatable}{theorem}{thmfold}\label{thm:fold}
There exists an algorithm that, given a string $x\in \Sigma^n$ and a parameter $\tau\in \Z_+$,
in $\Oh(n+\frac{n^3}{\tau^3})$ time computes a value $v$ such that $\rna(x) \le v \le \tau \rna(x)$.
\end{restatable}

Before describing the algorithm approximating $\rna(x)$ we state a few definitions and tools required to design the algorithm.
For a string $x\in \Sigma^n$, we say that a non-crossing matching $M\sub [n]^2$ is a \emph{folding alignment}
of $x$ if $x[j]=\rev{x[i]}$ holds for every $(i,j)\in M$. In this case, we also define its cost to be $\rna(x,M):=n-2|M|$.

\begin{fact}
For every string $x\in \Sigma^n$, the folding distance $\rna(x)$ is equal to the minimum cost $\rna(x,M)$
among all folding alignments $M$ of $x$.
\end{fact}
\begin{proof}
	We first show that, by induction on $n$, that $\rna(x)\le \rna(x,M)$ holds for every folding alignment $M$ of $x$.
	The claim is trivial if $n=0$.
	If $M(1)=\bot$, then we construct a folding alignment $M'=\{(i-1,j-1) : (i,j)\in M\}$ 
	of $x'=x[2\dd n]$. 
	By the inductive assumption, $\rna(x)\le \rna(x')+1\le \rna(x',M')+1 = (n-1)-2(|M|-1)+1= n-2|M|= \rna(x,M)$.
	If $M(1)=n$, then we construct a folding alignment $M'=\{(i-1,j-1) : (i,j)\in M\sm \{(1,n)\}\}$ of $x'=x[2\dd n)$.
	By the inductive assumption, $\rna(x)\le \rna(x')\le \rna(x',M') = (n-2)-2(|M|-1) = n-2|M|= \rna(x,M)$.
	Otherwise, we have $(1,p)\in M$ for some $p\in [2\dd n)$.
	In this case, we construct a folding alignment $M'=\{(i,j)\in M : j\le p\}$ of $x' = x[1\dd p]$,
	and a folding alignment $M'' = \{(i-p,j-p) : (i,j)\in M\text{ and }i>p\}$ of $x''=x(p\dd n]$.
	By the inductive assumption, $\rna(x)\le \rna(x')+\rna(x'')\le \rna(x',M')+\rna(x'',M'') = p-2|M'| + (n-p)-2|M''|=n-2|M|$; here, the last equality follows from the fact that any $(i,j)\in M$ with $i \le p$ and $j>p$ would violate the non-crossing property of $M$.
	
	As for the inverse inequality, we proceed by induction on $2\rna(x)+n$;
	again, the claim is trivial for $n=0$.
	If $x\in \RNA(\Sigma)$ and $x=ax'\bar{a}$ for $a\in \Sigma$ and $x'\in \RNA(\Sigma)$, then the inductive assumption yields a cost-$0$ folding alignment $M'$ of $x'$. In this case, we construct
	$M=\{(i+1,j) : (i,j)\in M'\}\cup\{(1,|x|)\}$ so that it is a cost-$0$ folding alignment of $x$.
	If $x = x'x''$ for non-empty $x',x''\in \RNA(\Sigma)$, then the inductive assumption yield cost-$0$ folding alignments $M',M''$ of $x'$ and $x''$, respectively.
	In this case, we set $M=M'\cup\{(i+|x'|,j+|x'|) : (i,j)\in M''\}$
	so that it is a cost-$0$ folding alignment of $x$.
	We henceforth assume that $\rna(x)>0$ and that deleting $x[p]$ is the first operation in the optimal sequence of character deletions transforming $x$ to a string in $\DYCK(\Sigma)$. The inductive hypothesis on the resulting string $x'$ yields a folding alignment $M'$ of $x$ with $\rna(x',M')=\rna(x')=\rna(x)-1$.
	We set $M = \{(i,j)\in M' : j < p\} \cup \{(i,j+1) : (i,j)\in M' \text{ and } i< p \le j\}\cup\{(i+1,j+1): (i,j)\in M'\text{ and }p \le i\}$ so that $M$ is a folding alignment of $x$ with $\rna(x,M)=n-2|M|=|x'|+1-2|M'|=\rna(x',M')+1=\rna(x)$.
	\end{proof}
	
Next we show folding edit distance satisfies triangular inequality.

\begin{lemma}\label{lem:rtriangle}
		All strings $x,y,z\in \Sigma^*$ satisfy $\rna(x\rev{z}) = \rna(x\rev{y}y\rev{z})\le \rna(x\rev{y})+\rna(y\rev{z})$. 
\end{lemma}
\begin{proof}
		The inequality follows from the fact that the folding language is closed under concatenations.
		As for the equality, we observe that it suffices to consider $|y|=1$: the case of $|y|=0$ is trivial,
		and the case of $|y|>1$ can be derived from that of $|y|=1$ by processing $y$ letter by letter.

		Next, observe that $\rna(x\rev{z}) \ge  \rna(x\rev{y}y\rev{z})$ because any folding alignment of $x\rev{z}$
		can be transformed to a folding alignment of $x\rev{y}y\rev{z}$ by matching $\rev{y}$ with $y$.

		It remains to prove $\rna(x\rev{z}) \le \rna(x\rev{y}y\rev{z})$.
		For this, we proceed by induction on $2\rna(x\rev{y}y\rev{z})+|x\rev{z}|$.
		If any optimum folding alignment of $x\rev{y}y\rev{z}$ leaves a character in $x$ or $\rev{z}$ unmatched,
		we apply the inductive assumption for an instance $(x',y,z')$ obtained from deleting this character: $\rna(x\rev{z}) \le \rna(x'\rev{z}')+1 \le \rna(x'\rev{y}y\rev{z}')+1 = \rna(x\rev{y}y\rev{z})$.
		If any optimum alignment of  $x\rev{y}y\rev{z}$ matches any two adjacent characters of $x$,
		any two adjacent characters of $z$, or the first character of $x$ with the last character of $z$,
		then we apply the inductive assumption for an instance $(x',y,z')$ obtained from deleting these two characters: $\rna(x\rev{z}) \le \rna(x'\rev{z}') \le \rna(x'\rev{y}y\rev{z}') = \rna(x\rev{y}y\rev{z})$.
		In the remaining case, all characters of $x$ and $\rev{z}$ must be matched to $\rev{y}$ or $y$, so $|x\rev{z}|\le 2$. 
		If $|x\rev{z}|\le \rna(x\rev{y}y\rev{z})$, then trivially $\rna(x\rev{z})\le \rna(x\rev{y}y\rev{z})$, so we may assume $\rna(x\rev{y}y\rev{z})<|x\rev{z}|$.
		Due to the fact that $\rna(x\rev{y}y\rev{z})\equiv |x\rev{y}y\rev{z}| \equiv |x\rev{z}| \pmod{2}$,
		this means that $\rna(x\rev{y}y\rev{z})=0$ and $|x\rev{z}|=2$.
		If $|x|=2$, then the optimum folding alignment is $\{(1,4),(2,3)\}$,
		so $x[1]=\rev{y}$ and $x[2]=\rev{\rev{y}}=y$, i.e., $\rna(x\rev{z})=\rna(x)=\rna(\rev{y}y)=0$.
		If $|\rev{z}|=2$, then the optimum folding alignment is $\{(1,4),(2,3)\}$,
		so $z[1]=\rev{y}$ and $z[2]=y$, i.e., $\rna(x\rev{z})=\rna(\rev{z})=\rna(\rev{y}\rev{\rev{y}})=0$.
		Finally, if $|x|=|\rev{z}|=1$, then the optimum folding alignment is $\{(1,2),(3,4)\}$;
		thus, $x=y=z$ and $\rna(x\rev{z})=0$.
\end{proof}

\newcommand{\Zz}{\mathbb{Z}_{\ge 0}}
\renewcommand{\S}{\mathcal{S}}

Next, we generalize the notion of a folding alignment.
\begin{definition}
Let $x\in \Sigma^n$ and let $\S$ be consistent window decomposition of $P\sub [1\dd n]$.
We say that $\S$ is a \emph{window alignment} of $x$ if $x[w']=\rev{x[w]}$ holds for every $(w,w')\in \S$.
Moreover, for $\rho\in \Zz$, we denote $\rna_\rho(x,\S)=n-|P|+\rho|\S|$.
We also define $\rna_\rho(x)$ as the minimum value $\rna_\rho(x,\S)$ among all window alignments $\S$ of $x$. 
\end{definition}

Observe that every window alignment can be interpreted as a folding alignment (by partitioning each window into individual characters) and every folding alignment can be interpreted as a window alignment (by considering length-1 windows). This yields the following observation:
\begin{observation}\label{obs:zero}
	Every string $x\in \Sigma^n$ satisfies $\rna_0(x)=\rna(x)$.
\end{observation}

% For a fixed string $x\in \Sigma^n$ and for $\rho\in \Zz$, we define a function $\D_\rho$ such that $\D_\rho(i,j)=\rna_\rho(x(i\dd j])$
% for $i,j\in [0\dd n]$ with $i\le j$.
The following lemma provides a recursive formula for $\rna_\rho(x)$:
\begin{lemma}\label{lem:rhodp}
Let $x\in \Sigma^n$ and $\rho\in \Zz$.
If $n\le 1$, then $\rna_\rho(x)=n$.
Otherwise, 
\[\rna_\rho(x) =
\min\begin{cases}
	\rna_\rho(x[1\dd p]) + \rna_\rho(x(p\dd n]) & \text{\!\!for }p\in [1\dd n),\\
	\rho+\rna_\rho(x(\ell\dd n-\ell]) & \text{\!\!for }\ell\in [1\dd \lfloor{\frac{n}{2}}\rfloor]\text{ such that } x[1\dd \ell]=\rev{x(n-\ell\dd n]}.
\end{cases}
\]
\end{lemma}
\begin{proof}
We proceed by induction on $n$.
If $n\le 1$, then the only window alignment of $x$ is the empty alignment.
Hence, $\rna_\rho(x)=\rna_\rho(x,\emptyset)=n$.

As for the lower bound on $\rna_\rho(x)$, we consider an optimal window alignment $\S$ of $x$.
If $\S$ does not contain any window pair involving $x[n]$, 
then $\S$ is also a window alignment of $x[1\dd n)$.
In particular, $\rna_\rho(x) = \rna_\rho(x,\S)=\rna_\rho(x[1\dd n),\S)+1 \ge \rna_\rho(x[1\dd n-1]) + \rna_\rho(x(n-1\dd n])$.
Otherwise, let $((p\dd p+\ell],(n-\ell\dd n])$ be the window pair involving $x[n]$.
If $p=0$, we define a window alignment $\S'$ of $=x(\ell\dd n-\ell]$ so that all the remaining window pairs in $\S$ are shifted by $\ell$ positions to the left.
We then have $\rna_\rho(x) = \rna_\rho(x,\S)= \rho + \rna_\rho(x(\ell\dd n-\ell],\S')\ge \rho + \rna_\rho(x(\ell\dd n-\ell])$. Moreover, $x[1\dd \ell]=\rev{x(n-\ell\dd n]}$ follows from the definition of a window alignment.
In the remaining case of $p>0$, we define window alignments $\S',\S''$ of $x[1\dd p]$ and $x(p\dd n]$, respectively, so that all the window pairs in $\S$ are either preserved in $\S'$ or shifted by $p$ positions to the left and inserted to $\S''$.
We then have $\rna_\rho(x)=\rna_\rho(x,\S)=\rna_\rho(x[1\dd p],\S')+\rna_\rho(x(p\dd n],\S'')\ge \rna_\rho(x[1\dd p])+\rna_\rho(x(p\dd n])$.

As for the upper bound on $\rna_\rho(x)$, let us first consider
$\rna_\rho(x[1\dd p])+\rna_\rho(x(p\dd n])$ for some $p\in [1\dd n)$.
Consider optimal window alignments $\S'$ of $x[1\dd p]$ and $\S''$ of $x(p\dd n]$, respectively.
Taking all the window pairs in $\S'$ and shifting all the window pairs in $\S''$ by $p$ positions to the right,
we obtain a window alignment $\S$ of $x$ that satisfies
$\rna_\rho(x)\le \rna_\rho(x,\S)=\rna_\rho(x[1\dd p],\S')+\rna_\rho(x(p\dd n],\S'') = \rna_\rho(x[1\dd p])+\rna_\rho(x(p\dd n])$.
Next, consider $\ell\in [1\dd \floor{\frac{n}{2}}]$ such that $x[1\dd \ell]=\rev{x(n-\ell\dd n)}$. 
Consider an optimal window alignment $\S'$ of $x(i+\ell\dd j-\ell]$.
Taking all the window pairs in $\S'$ shifted by $\ell$ positions to the right,
plus a window pair $([1\dd \ell],(n-\ell\dd n])$, we obtain a window alignment $\S$ of $x$ that satisfies  
$\rna_\rho(x)\le \rna_\rho(x,\S)=\rho+\rna_\rho(x(\ell\dd n-\ell],\S')=\rho+\rna_\rho(x(\ell\dd n-\ell])$.
\end{proof}

We say that a string $x$ is \emph{irreducible} if $x[i+1]\ne \rev{x[i]}$ holds for all $i\in [1\dd |x|)$.
Due to the equality in \cref{lem:rtriangle}, the task of computing or approximating $\rna(x)$
easily reduces to the special case when $x$ is irreducible (subsequent characters $x[i+1]=\rev{x[i]}$
can be removed without affecting the folding edit distance).

\begin{lemma}\label{lem:rho}
For every $x\in \Sigma^n$ and $\rho\in \Zz$, we have $\rna_\rho(x)\ge \rna(x)$.
Moreover, if $x$ is non-empty and irreducible, then $\rna_\rho(x) \le (1+2\rho)\rna(x)-\rho$.
\end{lemma}
\begin{proof}
As for the first inequality, consider the optimal window alignment $\S$ of $x$.
We construct $M=\bigcup_{(w,w')\in \S} \{(s(w)+i, e(w')-i) : i\in [0\dd |w|)\}$.
It is easy to see that $M$ is a non-crossing matching and that $\rna(x,M)=\rna_\rho(x,\S)-\rho|\S|$.
Hence, $\rna(x) \le \rna(x,M) \le \rna_\rho(x,\S)=\rna_\rho(x)$.

As for the second inequality, we proceed by induction on $|x|$.
However, we prove a stronger inequality $\rna_\rho(x) \le (1+2\rho)\rna(x)-2\rho$
when $x$ admits an optimal folding alignment that does not contain $(1,n)$.
In the base case of $n=1$, we have $\rna_\rho(x) = 1 = (1+2\rho)-2\rho = (1+2\rho)\rna(x)-2\rho$.

If $\rna(x)=\rna(x[1\dd p]) + \rna(x(p\dd n])$ for some $p\in (1\dd n)$,
then we consider $x'=x[1\dd p]$ and $x''=x(p\dd n]$.
By \cref{lem:rhodp} and the inductive assumption,
$\rna_\rho(x) \le \rna_\rho(x')+\rna_\rho(x'')\le (1+2\rho)\rna(x')-\rho+(1+2\rho)\rna(x'')-\rho
=(1+2\rho)\rna(x)-2\rho$.
Otherwise, let $M$ be an optimal folding alignment of $x$ and let $\ell$ be the largest integer such that $M \supseteq \{(p,n-p) : p\in [1\dd \ell]\}$.
The latter condition implies $x(n-\ell\dd n]=\rev{x[1\dd \ell]}$, whereas irreducibility
of $x$ guarantees $\ell \le \lfloor \frac{n}{2}\rfloor$.
In this case, \cref{lem:rhodp} yields $\rna_\rho(x)\le 1+\rna_\rho(x(\ell\dd n-\ell]) \le \rho+(1+2\rho)\rna(x(\ell\dd n-\ell])-2\rho = (1+2\rho)\rna(x)-\rho$.
\end{proof}

% For $\rho\ge 1$, our algorithm works in $\Oh(n+(n/\rho)^3)$ time and provides an $\Oh(1)$-factor approximation of $\rna_\rho(x)$; this yields an $\Oh(\rho)$-factor approximation of $\rna(x)$.

% For this, we decompose $x$ into blocks of length $\rho$ (the final block might be shorter).

\begin{algorithm}
	$m:=\lfloor{n/s}\rfloor$\;
	\For{$a=m$ \KwSty{down to} $1$}{
		$D[a,a]:=0$\;
		\lIf{$a<m$}{$D[a,a+1]:=s$}
		\For{$b:=a+2$ \KwSty{to} $m$}{
			$D[a,b]=\infty$\;
			\For{$d : = 1$ \KwSty{to} $\floor{\frac{b-a-2}{2}}$}{
				\If{$\rev{x((b-d)s\dd bs]}\text{ is a substring of }x(as\dd (a+d+2)s]$}{\label{ln:ipm}
					$D[a,b] :=\min(D[a,b],12s+D[a+d+2,b-d])$\;
				}
			}
			\For{$c=a+1$ \KwSty{to} $b-1$}{
				$D[a,b]:=\min(D[a,b], D[a,c]+D[c,b])$\;
			}
		}
	}
	\Return{$D[0,m]+n\bmod s$}
	\caption{Approximate folding distance.}\label{alg:fold}
\end{algorithm}

The algorithm for approximating folding edit distance takes a string $x$ and a parameter $s\in \Z_{+}$ as input and outputs a matrix $D$ where $D[\lfloor \frac{i}{s}\rfloor, \lfloor \frac{j}{s}\rfloor]$ provides an estimation of the distance of $x(i\dd j]$. To compute an entry $D[a,b]$, the algorithm searches for substrings $x((b-d)s\dd bs]$, where $d\in[1\dd \lfloor \frac{b-a-2}{2}\rfloor]$, (the substring starts and ends at indices divisible by $s$) whose inverses occur in $x(as\dd (a+d+2)s]$; if such a match exists,
the algorithm sets $D[a,b]:=\min(D[a,b],12s+D[a+d+2,b-d])$. Next the algorithm considers all possible pivots $c\in(a\dd b)$ and sets $D[a,b]:=\min(D[a,b], D[a,c]+D[c,b])$.

\begin{lemma}\label{lem:fold}
Consider running \cref{alg:fold} on a string $x\in \Sigma^n$ with parameter $s\in \Z_{+}$.
For every fragment $x(i\dd j]$ of $x$, we have $\frac13 D[\floor{\frac{i}{s}},\floor{\frac{j}{s}}] \le \rna_{8s}(x(i\dd j])+i\bmod s - j\bmod s \le D[\floor{\frac{i}{s}},\floor{\frac{j}{s}}]$.
In particular, the returned value is between $\rna_{8s}(x)$ and $3\rna_{8s}(x)$.
\end{lemma}
\begin{proof}
We proceed by induction on $j-i$.
In the base case of $j-i\in \{0,1\}$,
we have $D[\floor{\frac{i}{s}},\floor{\frac{j}{s}}]= \floor{\frac{j}{s}}-\floor{\frac{i}{s}} = j-j\bmod s-(i-i\bmod s) = (j-i)+i\bmod s - j\bmod s = \rna_{8s}(x(i\dd j])+i\bmod s - j\bmod s$.

As for the inductive step, let us denote $a=\floor{\frac{i}{s}}$
and $b=\floor{\frac{j}{s}}$. 
We start with the upper bound on $\rna_{8s}(x(i\dd j])$.
If $D[a,b]=D[a,c]+D[c,b]$ holds for some $c\in (a\dd b)$,
then \cref{lem:rhodp} and the inductive assumption yield
\begin{multline*}
\rna_{8s}(i,j)\le \rna_{8s}(i,cs)+\rna_{8s}(cs,j) \le D[a,c]-i\bmod s + D[c,b]+j\bmod s \\ = D[a,b]+j\bmod s - i\bmod s.\end{multline*}
Next, suppose that $D[a,b]=12s+D[a+d+2,b-d]$ holds for some $d\in [1\dd \floor{\frac{b-a-2}{2}}]$
such that $\rev{x((b-d)s\dd bs]}$ is a substring of $x(as\dd (a+d+2)s]$.
The latter condition implies that $x(i'\dd i'+ds]=\rev{x((b-d)s\dd bs]}$ holds for some $i'\in [as\dd (a+2)s]$.
Then, \cref{lem:rhodp} and the inductive assumption yield
\begin{align*}\rna_{8s}(x(i\dd j]) &\le \rna_{8s}(x(i'\dd bs]) + |i'-i| + (j-bs)\\
&\le \rna_{8s}(x(i'+ds\dd (b-d)s]) + 8s + |i'-i|+ j\bmod s\\
&\le \rna_{8s}(x((a+d+2)s\dd (b-d)s]) + ((a+2)s-i') + 8s+ |i'-i| + j\bmod s\\
&\le D[a+d+2,b-d] + 8s + \max((a+2)s-i, (a+2)s+i-2i') + j\bmod s\\
&= D[a,b] - 4s + \max(2s-i\bmod s, 2s+i\bmod s) + j\bmod s\\
&= D[a,b] - 2s +  i\bmod s + j \bmod s \\
&< D[a,b] -i\bmod s + j\bmod s.\end{align*}

It remains to prove the lower bound on $\rna_{8s}(x(i\dd j])$.
If $\rna_{8s}(x(i\dd j])=\rna_{8s}(x(i\dd k])+\rna_{8s}(x(k\dd j])$ holds for some $k\in (i\dd j)$,
then the inductive assumption yields 
\begin{align*}D[a,b]
&\le D[a,\floor{\tfrac{k}{s}}]+D[\floor{\tfrac{k}{s}},b] \\
&\le 3(\rna_{8s}(x(i\dd k])+i\bmod s - k \bmod s)+3(\rna_{8s}(x(k\dd j])+k\bmod s - j \bmod s)\\
&= 3(\rna_{8s}(x(i\dd j])+i\bmod s - j\bmod s).\end{align*}
Otherwise, \cref{lem:rhodp} implies that $\rna_{8s}(x(i\dd j])=8s + \rna_{8s}(x(i+\ell\dd j-\ell])$
holds for some $\ell\in [4s\dd \floor{\frac{j-i}{2}}]$ such that $x(i\dd i+\ell]=\rev{x(j-\ell\dd j]}$.
Let $d := \floor{\frac{i+\ell}{s}}-a-2$ and observe that $d \ge \frac{\ell-3s}{s}\ge 1$
and $d \le \frac{\ell-s}{s}$.
Moreover, $(a+d+2)s \le i+\ell \le j-\ell < (b+1)s-\ell \le (b-d)s$,
so $d\le \floor{\frac{b-a-2}{2}}$,
and $(a+d+2)s-i=(d+2)s-i\bmod s > (d+1)s > ds + j\bmod s = j-(b-d)s$,
so $\rev{x((b-d)s\dd j]}$ is a prefix of $x(i\dd (a+d+2)s]$ and thus $\rev{x((b-d)s\dd bs]}$ is a substring of $x(as\dd (a+d+2)s]$.
Consequently, the inductive assumption yields 
\begin{align*}D[a,b]&\le 12s + D[a+d+2,b-d]\\
&\le 3(\rna_{8s}(x((a+d+2)s\dd (b-d)s])+4s) \\
&\le 3(\rna_{8s}(x(i+\ell \dd j-\ell]) + (i+\ell)-(a+d+2)s + (b-d)s-(j-\ell)+4s)\\
&= 3(\rna_{8s}(x(i \dd j]) + (i-as)-(j-bs)+2(\ell-3s-ds))\\
&\le 3(\rna_{8s}(x(i \dd j])+i\bmod s - j\bmod s).\qedhere\end{align*}\end{proof}

\begin{fact}\label{fct:time}
\cref{alg:fold} can be implemented in $\Oh(n+\frac{n^3}{s^3})$ time.
\end{fact}
\begin{proof}
Not accounting for \cref{ln:ipm}, \cref{alg:fold} already works in $\Oh(\frac{n^3}{s^3})$.
As for \cref{ln:ipm}, the task is to verify whether a fragment of $\rev{x}$ of length $ds$ has an occurrence 
within a fragment of $x$ of length $(d+2)s$. This is precisely the problem solved by the \textsc{Internal Pattern Matching} queries of~\cite{KRRW15,phd} on the text $x\bar{x}$. After linear-time preprocessing of the text,
\textsc{IPM} queries can be answered in time proportional to the ratio of the lengths of the two fragments involved,
which is $\frac{(d+2)s}{ds}\le 3$ in our setting.
Hence, the total running time is $\Oh(n+\frac{n^3}{s^3})$, as claimed.
\end{proof}

\thmfold*
\begin{proof}
If $\tau \le 51$, we use the classic algorithm computing $\rna(x)$ exactly in $\Oh(n^3)$ time (it also follows from \cref{lem:rhodp,obs:zero}). Thus, we henceforth assume that $\tau \ge 51$.

In the preprocessing step, we reduce $x$ by exhaustively removing adjacent characters $x[i]x[i+1]$
such that $x[i+1]=\rev{x[i]}$. This operation preserves $\rna(x)$ due to the equality in \cref{lem:rtriangle}.
Moreover, it can be implemented in $\Oh(n)$ time via a left-to-right scan of $x$
that maintains an auxiliary string $y$ representing the reduced version of the already processed prefix $x[1\dd p]$.
While scanning the subsequent character $x[p+1]$, we check whether $\rev{x[p+1]}$ matches the last character of $y$.
Depending on the result, we either remove the last character of $y$ or append $x[p+1]$ to $y$.

After this $\Oh(n)$-time preprocessing, we are left with the task of approximating $\rna(x)$
for an irreducible string $x$. If $x$ is empty, we return $\rna(x)=0$.
Otherwise, we run \cref{alg:fold} with $s = \floor{\frac{\tau-3}{48}}\in \Z_+$ and return its result.
By \cref{fct:time}, this call costs $\Oh(n+\frac{n^3}{s^3})=\Oh(n+\frac{n^3}{\tau^3})$ time,
Moreover, \cref{lem:fold,lem:rho} guarantee that the resulting value $v$ satisfies
\[\rna(x) \le \rna_{8s}(x) \le v \le 3\rna_{8s}(x) \le 3((1+16s)\rna(x)-8s) < (3+48s)\rna(x)\le \tau \rna(x).\qedhere\]
\end{proof}

\bibliographystyle{alphaurl}
\bibliography{main}

\end{document}